\documentclass[11pt,draftclsnofoot,onecolumn]{IEEEtran}

\makeatletter
\def\ps@headings{%
\def\@oddhead{\mbox{}\scriptsize\rightmark \hfil \thepage}%
\def\@evenhead{\scriptsize\thepage \hfil \leftmark\mbox{}}%
\def\@oddfoot{}%
\def\@evenfoot{}}
\makeatother

\usepackage{mathrsfs}
\usepackage{times}
\usepackage{amsmath,amssymb,latexsym}
\usepackage{epsfig}
\usepackage{tabularx}
\usepackage{ifpdf}
\usepackage{cite}
\usepackage{algorithmic}
\usepackage{algorithm}
\usepackage{array}
\usepackage{mdwmath}
\usepackage{mdwtab}
\usepackage{eqparbox}
\usepackage{subfigure}

\usepackage{slashbox}
\usepackage{url}
\usepackage{textcomp}
\usepackage{color}
\usepackage{subfigure}
\usepackage{graphicx}

\newtheorem{proposition}{\bf Proposition}

\newtheorem{lemma}{{\bf Lemma}}
\newtheorem{assumption}{Assumption}

\newtheorem{remark}{Remark}
\newtheorem{definition}{Definition}

\begin{document}

\IEEEoverridecommandlockouts

\title{Joint Scheduling and Power Allocations for Traffic Offloading via Dual-Connectivity}

\vspace{0.2in}

\author{
\IEEEauthorblockN{Yuan Wu, Yanfei He, Liping Qian, Jianwei Huang, Xuemin (Sherman) Shen}
\thanks{Y. Wu and Y. He are with College of Information Engineering, Zhejiang University of Technology, Hangzhou, China, (email: iewuy@zjut.edu.cn).}
\thanks{L. Qian is with College of Computer Science and Technology, Zhejiang University of Technology, Hangzhou, China, (email: lpqian@zjut.edu.cn).}
\thanks{J. Huang is with the Network Communications and Economics Lab, Department of Information Engineering, The Chinese University of Hong Kong, Hong Kong (e-mail: jwhuang@ie.cuhk.edu.hk).}
\thanks{S. Shen is with the Department of Electrical and Computer Engineering, University of Waterloo, Waterloo, ON N2L 3G1, Canada (e-mail: xshen@bbcr.uwaterloo.ca)}
\vspace{-0.5in}
}

\IEEEoverridecommandlockouts
\maketitle
\begin{abstract}
With the rapid growth of mobile traffic demand, a promising approach to relieve cellular network congestion is to offload users' traffic to small-cell networks. In this paper, we investigate how the mobile users (MUs) can effectively offload traffic by taking advantage of the capability of dual-connectivity, which enables an MU to simultaneously communicate with a macro base station (BS) and a small-cell access point (AP) via two radio-interfaces. Offloading traffic to the AP usually reduces the MUs' mobile data cost, but often at the expense of suffering increased interferences from other MUs at the same AP. We thus formulate an optimization problem that jointly determines each MU's traffic schedule (between the BS and AP) and power control (between two radio-interfaces). The system objective is to minimize all MUs' total cost, while satisfying each MU's transmit-power constraints through proper interference control. In spite of the non-convexity of the problem, we design both a centralized algorithm and a distributed algorithm to solve the joint optimization problem. Numerical results show that the proposed algorithms can achieve the close-to-optimum results comparing with the ones achieved by the LINGO (a commercial optimization software), but with significantly less computational complexity. The results also show that the proposed adaptive offloading can significantly reduce the MUs' cost, i.e., save more than 75\% of the cost without offloading traffic and 65\% of the cost with a fixed offloading.
\end{abstract}

\IEEEpeerreviewmaketitle

\section{Introduction}
\label{section_introduction}
With the rapid growth of smart handheld devices and mobile internet services, traffic demand in cellular networks has been growing tremendously \cite{WhitePaper:Cicso}, which causes frequent congestions and imposes a heavy burden on network operators. A cost effective approach to relieve network congestion is to offload the traffic of mobile users (MUs) to spatially spread small-cell networks, e.g., femtocells and WiFi access points (APs) \cite{Paper:SurveyOffloading}\cite{Paper:SurveyOffloading2}. From the network operators' perspective, offloading traffic can effectively exploit the additional network capacities provided by the multi-tiers small cells and reduce the needs of costly and time-consuming network infrastructure upgrade. From the MUs' perspective, offloading traffic reduces their costs, since small cells usually offer a low price for mobile data than cellular operators. Traffic offloading becomes increasing attractive, as an increasing percentage of mobile devices are equipped with multiple radio-interfaces that facilitate to flexibly simultaneous connections to multiple networks \cite{Papre:Lee,Paper:DTN_arch,Paper:Hossain}.

In particular, a new paradigm of \textit{small-cell dual-connectivity} is gaining momentum in both industry practice \cite{WhitePaper:Qualcomm}\cite{WhitePaper:NSN} and the 3GPP LTE-A standardizing activities \cite{Paper:LTEA}\cite{Paper:Standard}. Through the dual-connectivity, each MU can simultaneously communicate with a macro base station (BS) and a small-cell AP via two different radio-interfaces. This enables the MUs to flexibly schedule their traffic between two networks to  achieve efficient traffic offloading. For instance, an MU can schedule its delay-sensitive small-volume data traffic (e.g., voice-over-IP traffic) to the macro BS, and offload its delay-tolerant large-volume traffic (e.g., file downloading/uploading) to a small-cell AP at the same time.

Nevertheless, different from the centralized resource management and orthogonal channel allocation in cellular networks, a small-cell AP (such as WiFi) often allows multiple MUs to share the same channels in a distributed fashion, which leads to significant mutual interferences among the MUs. It means that if each MU aggressively offloads its traffic to the AP, then each MU needs to spend significant amount of transmit-power to the AP to combat the mutual interferences, which might significantly reduce the benefit from offloading traffic.

Several prior studies have taken such interferences into consideration when designing data offloading algorithms \cite{Paper:Kang}\cite{Paper:Ho}\cite{Paper:EE_Chen}. However, the interference-aware traffic offloading design with dual-connectivity is still an open problem. This problem becomes especially complicated if we consider different transmit-power constraints at each MU's different radio-interfaces \cite{Paper:PowerNI}\cite{Paper:PowerDBM}, which require the MU to carefully allocate its transmit-powers to match the need of the scheduled traffic. To tackle the above challenging problem, we propose a joint optimization of the traffic scheduling and transmit-power allocations with the MUs' dual-connectivity. Our contributions can be summarized as follows.
\begin{itemize}
\item \textit{Novel Joint Optimization Formulation:} We formulate a cost minimization problem, in which each MU jointly determines its traffic scheduling to the AP and BS and the transmit-powers at the two radio-interfaces. The objective is to minimize the total cost of all MUs, while meeting each MU's traffic demand and its transmit-power constraints. To the best of our knowledge, such a formulation targeted for the MUs' traffic offloading with dual-connectivity has not been studied before.

\item \textit{Centralized Algorithm Design:} We first propose a centralized algorithm for solving the joint traffic scheduling and power control problem. In spite of non-convexity of the problem, we perform a series of manipulations to transform the joint optimization into an equivalent SINR-assignment problem (here SINR denotes the signal-to-interference-plus-noise ratio). We then explore the \textit{hidden monotonicity} of the SINR-assignment problem and design a two-layered algorithm to solve it. Based on the obtained SINRs at the AP, we derive the MUs' traffic scheduling and transmit-powers to minimize the total cost.

\item \textit{Distributed Algorithm Design:} We next propose a distributed algorithm to solve the joint optimization problem, by focusing on a practically important case that the channel bandwidth of the AP is no smaller than the BS's allocated bandwidth. In this case, we identify \textit{the hidden concave minimization property} of the SINR-assignment problem and propose a distributed algorithm to compute each MU's SINR at the AP. Based on the obtained SINRs at the AP, we derive the MUs' traffic scheduling and transmit-powers to minimize the total cost in a distributed manner.

\item \textit{Performance Improvement:} Numerical results show that both the proposed algorithms can achieve the close-to-optimum results which are obtained by the LINGO (a commercial optimization software \cite{Software:LINGO}) but with a significant less computational complexity. The numerical results also validate that the proposed traffic offloading can significantly reduce the MUs' cost, namely, saving more than 75\% of the cost without performing any offloading and more than 65\% of the cost with a fixed offloading scheme. 
\end{itemize}

The rest of this paper is organized as follows. Section \ref{section_relatedworks} describes the related studies. Section \ref{section_model} presents the problem formulation. Section \ref{section_transformation} presents a series of transformations that facilitate the following algorithm designs. Sections \ref{section_centralized} and \ref{section_distributed} propose the centralized and distributed algorithms, respectively, to solve the problem. Section \ref{section_numerical} presents the numerical results, and we conclude this study in Sect. \ref{section_conclusion}.

\vspace{0.1in}
\section{Related Works}
\label{section_relatedworks}
Since the seminal studies \cite{Papre:Lee}\cite{Paper:DTN_arch}, there have been many studies that investigated traffic offloading via infrastructure-based small-cell networks\footnote{Besides offloading traffic via infrastructure-based networks, it is also possible to offload traffic through peer-to-peer communications (e.g., the device-to-device communications \cite{Paper:D2D_Lei}). This, however, is not the focus of this study.}. They can be roughly categorized into two groups as follows:
\begin{itemize}
\item \textit{Network-oriented traffic offloading}. The first group of studies mainly focused on optimizing traffic offloading from the networks' perspectives. In \cite{Paper:Ho}, Ho \textit{et al.} considered the inter-cell interference when accommodating the offloaded traffic. Taking into account the coupling effect due to the mutual interference, the authors formulated a utility optimization framework for distributing traffic loads among different macro-cells. In \cite{Paper:EE_Chen}, Chen \textit{et al.} also considered the interference among different small cells when offloading traffic, and proposed a framework that facilitates the macro BS to manage the small cells to minimize the energy consumption. In \cite{Paper:DoubleAuction_Iosifidis}, Iosifidis \textit{et al.} considered the coupled capacities of different APs due to interference, and proposed a double auction mechanism that matches the network operators with the most suitable APs for data offloading. In \cite{Paper:Backhaul_Duan}, considering the suffered interference due to serving macro-cell MUs, Yang \textit{et al.} proposed a refund mechanism for network operator to motivate the small-cell APs to admit macro-cell MUs for traffic offloading.

\item \textit{User-oriented traffic offloading}. The second group of studies focused on optimizing the traffic offloading from the MUs' perspectives, with the key issues including interference management and delay-offloading tradeoff. For instance, in \cite{Paper:Kang}, Kang \textit{et al.} considered the scenario of one BS and one third-party WiFi AP, and investigated how different MUs choose either the BS or the AP for offloading traffic. Considering the mutual interference among different MUs, the authors formulated the networks-selection problem as a binary nonlinear programming problem for maximizing the system-wise reward.
    Studies in \cite{Paper:Delay_Zhuo}\cite{Paper:EconomicGain_Song} proposed several different schemes to motivate the MUs to delay their traffic offloading by leveraging the delay tolerance and better future network conditions. In \cite{Paper:AMUSE}, Im \textit{et al.} proposed an MU-centric cost-aware WiFi offloading system, which considers the MU's throughput-delay tradeoff and cost budget to decide when to offload which type of traffic.
\end{itemize}

In spite of the above studies, only few studies investigated the paradigm of small-cell dual-connectivity (which has been gaining momentum in the industry practice \cite{WhitePaper:Qualcomm}\cite{WhitePaper:NSN} and the 3GPP LTE-A standardizing activities \cite{Paper:LTEA}\cite{Paper:Standard}) for offloading the MUs' traffic. In \cite{Paper:Dual}, Jha \textit{et al.} discussed several technical challenges and potential solution directions regarding the small-cell dual-connectivity in cellular networks. However, to the best of our knowledge, there exists no existing study that investigated the optimal management of the MUs' traffic scheduling and transmit-power allocations for traffic offloading via the dual-connectivity.

We emphasize that although there are several studies investigating how the traffic offloading can benefit the network operators, it is less understood about how much the MUs can benefit from offloading traffic in terms of saving the mobile data cost. Our study sheds light on this benefit and illustrates how to optimize such benefit through both centralized and distributed approaches.

\section{System Model and Problem Formulation}
\label{section_model}

\subsection{System Model}
\begin{figure}[tbph]
  \centering
  \includegraphics[scale=0.7]{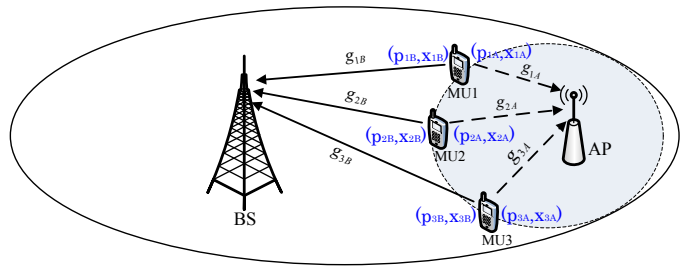}\\
  \vspace{-0.2in}
  \caption{Illustration of the system model.} \label{fig_system}
\end{figure}

Figure \ref{fig_system} illustrates the system model, where a macro BS is serving a set of MUs $\mathcal{I}=\{1,2,...,I\}$ which perform the uplink transmissions\footnote{We focus on traffic offloading in uplink case, where we need to consider the MUs' limited resources. Such an issue makes the uplink resource allocation more challenging than the downlink case. This is also motivated by the rapid growth of user-generated contents (such as user-generated videos on social networks), which has significantly increased the MUs' uplink traffic volume.}. There also exists a small-cell AP that can provide traffic offloading services to the MUs through dual-connectivity. Each MU $i$ has two radio-interfaces, one for sending traffic to the BS and one for offloading traffic to the AP. We use $x_{iA}$ and $x_{iB}$ to denote MU $i$'s transmission rates to the AP and the BS, respectively. We use $p_{iA}$ and $p_{iB}$ to denote MU $i$'s transmit-powers to the AP and the BS, respectively\footnote{In future work, we will extend our model to the case of multiple APs, where an MU selects which AP to offload traffic to.}. In the rest of this paper, the subscripts ``A" and ``B" denote ``AP" and ``BS", respectively. We consider that the macro BS and AP operate on different spectrums to provide service\footnote{The small cell may operate on a separated licensed spectrum (such as the ``separate carrier" scheme for femtocell \cite{Paper:APChannel1}) or on a separate unlicensed spectrum (such as the case of WiFi AP).}.

The small-cell AP allows multiple MUs to share the same spectrum (channel) \cite{Paper:spectrumsharing1}\cite{Paper:SpectrumSharing_Niyato}, which results in mutual interference among the MUs when offloading traffic. To make the discussions more concrete, we adopt the throughput model under the interference channel as \cite{Paper:Kang}\cite{Paper:EE_Chen}\cite{Paper:Backhaul_Duan}, i.e., given the MUs' transmit-powers $\{p_{iA}\}_{i\in\mathcal{I}}$, MU $i$'s transmission rate to the AP is
\begin{eqnarray}
x_{iA}=W \log_2\left(1+\frac{p_{iA}g_{iA}}{\sum_{j\neq i,j\in\mathcal{I}}p_{jA}g_{jA}+n_A}\right), \forall i\in\mathcal{I}, \label{eq_rate_AP}
\end{eqnarray}
where $W$ denotes the AP's channel bandwidth, and $g_{iA}$ denotes the channel gain from MU $i$ to the AP. $n_A=W n_0$ denotes the power of the background noise at the AP, with $n_0$ denoting the power density.

The BS allocates orthogonal sub-channels to different MUs for accommodating the uplink traffic (such as in OFDMA). Given MU $i$'s transmit-power $p_{iB}$, MU $i$'s uplink transmission rate to the BS is
\begin{eqnarray}
x_{iB}=B \log_2\left(1+\frac{p_{iB}g_{iB}}{n_B}\right), \forall i\in\mathcal{I}, \label{eq_rate_BS}
\end{eqnarray}
where $g_{iB}$ is the channel power gain from MU $i$ to the BS, and $B$ is the bandwidth allocated to MU $i$ by the BS (we consider that the BS's channel allocation is given in this study). $n_B=B n_0$ denotes the power of the background noise at the BS.

Each MU $i$ needs to satisfy a traffic demand $R_i^{\text{req}}$ through the transmissions to both BS and AP, i.e.,
\begin{eqnarray}
x_{iA}+x_{iB}\geq R_i^{\text{req}}, \forall i\in\mathcal{I}. \label{eq_trafficdemand}
\end{eqnarray}

Similar to \cite{Paper:Kang}\cite{Paper:Backhaul_Duan}\cite{Paper:Delay_Huang2}, we consider
usage-based pricing schemes by both the AP and the BS. Let $\pi_A$ and $\pi_B$ to denote the unit-prices announced by the AP and BS, respectively. Then, MU $i$'s total transmission cost in one unit time is
\begin{eqnarray}
C_{i}(x_{iA},x_{iB})=\pi_A x_{iA}+\pi_B x_{iB}, \forall i\in\mathcal{I}.\label{eq_MUcost}
\end{eqnarray}

\subsection{Problem Formulation}
We are interested in jointly optimizing all MUs' traffic scheduling $\{x_{iA},x_{iB}\}_{i\in\mathcal{I}}$ and the transmit-powers $\{p_{iA},p_{iB}\}_{i\in\mathcal{I}}$, to minimize the total cost. Here, ``CMP" stand for ``Cost Minimization Problem".
\begin{eqnarray}
\text{(CMP):} && \text{Minimize} \sum_{i\in\mathcal{I}}C_i(x_{iA},x_{iB})=\sum_{i\in\mathcal{I}}\pi_A x_{iA}+\sum_{i\in\mathcal{I}}\pi_B x_{iB} \nonumber\\
&& \text{Subject to: } 0\leq p_{iA}\leq P_{iA}^{\max},\forall i\in\mathcal{I},\label{P1_Con_PiAmax}
\end{eqnarray}
\begin{eqnarray}
&& \text{~~~~~~~~~~~~~~}0\leq p_{iB}\leq P_{iB}^{\max},\forall i\in\mathcal{I},\label{P1_Con_PiBmax}\\
&& \text{~~~~~~~~~~~~~~}p_{iA}+p_{iB}\leq P_i^{\max},\forall i\in\mathcal{I},\label{P1_Con_Pimax}\\
&& \text{~~~~~~~~~~~~~~}\text{Constraints } (\ref{eq_rate_AP}),(\ref{eq_rate_BS}),\text{ and } (\ref{eq_trafficdemand}), \nonumber\\
&& \text{Variables: } (x_{iA},x_{iB}),\forall i\in\mathcal{I} \text{ and } (p_{iA},p_{iB}),\forall i\in\mathcal{I}. \nonumber
\end{eqnarray}
Constraints (\ref{P1_Con_PiAmax})-(\ref{P1_Con_Pimax}) are motivated by the fact that different radio-interfaces can have different restrictions on the power consumptions \cite{Paper:PowerNI}\cite{Paper:PowerDBM}. Constraint (\ref{P1_Con_PiAmax}) ensures that MU $i$'s transmit-power $p_{iA}$ to the AP cannot exceed the upper-bound $P_{iA}^{\max}$. Constraint (\ref{P1_Con_PiBmax}) ensures that MU $i$'s transmit-power $p_{iB}$ to the BS cannot exceed the upper-bound $P_{iB}^{\max}$. Constraint (\ref{P1_Con_Pimax}) ensures that MU $i$'s total power consumption at the two interfaces cannot exceed the upper bound $P_i^{\max}$.

The key challenge to solve Problem (CMP) is due to the intrinsic non-convexity of (\ref{eq_rate_AP}). In the rest of this paper, we will develop efficient algorithms to solve Problem (CMP) to achieve close-to-optimum performance. Before presenting the details, we make the following assumption.
\begin{assumption}
\label{assumption_feasible} Problem (CMP) is feasible.
\end{assumption}

Assumption \ref{assumption_feasible} can be satisfied by imposing an admission control policy that selects an appropriate group of MUs to serve, such that the MUs' traffic demands can be satisfied within their transmission power constraints\footnote{A heuristic approach to perform admission control is that the MU (let us say MU $i$) with $(2^{\frac{R_i^{\text{req}}}{B}}-1)\frac{n_B}{g_{iB}}\leq P_{iB}^{\max}$ is admitted.}. Assumption \ref{assumption_feasible} enables us to focus on evaluating the benefit of the traffic offloading via dual-connectivity and designing algorithms to achieve this benefit.

Before proposing the algorithms to solve Problem (CMP), we will first present a series of problem formulations equivalent to Problem (CMP). Specifically, we use Fig. \ref{Figure_transformation} to show how these problem formulations are related and where they are located in this paper.

\begin{figure}[tbph]
\centering
\includegraphics[scale=0.7]{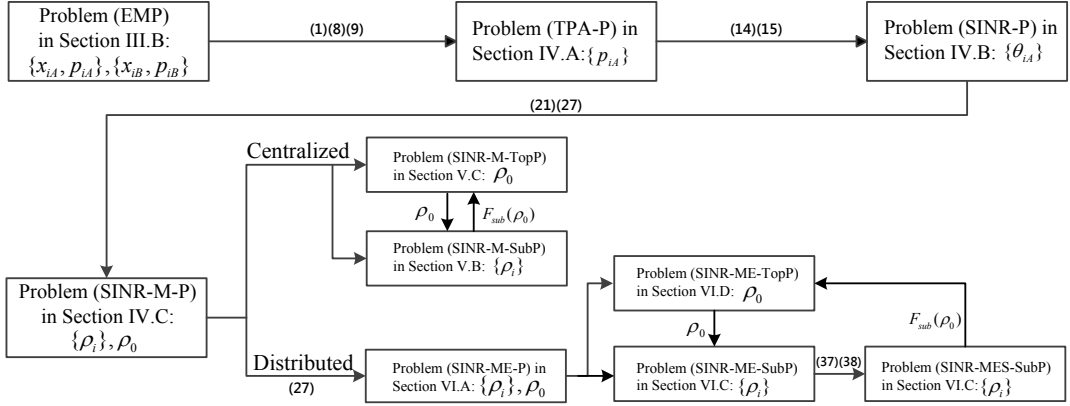}
\caption{Connections among the problem formulations. We also mark out the decision variables of each problem formulation.}
\label{Figure_transformation}
\end{figure}

\section{Equivalent Transformations of Problem (CMP)}
\label{section_transformation}
This section presents a series of problem formulations that facilitate our later algorithm designs.

\subsection{Transforming Problem (CMP) into a Transmit-Power Allocation Problem}
We first identify the following property of Problem (CMP).
\begin{lemma}
\label{lemma_binding}
Constraint (\ref{eq_trafficdemand}) is tight at any optimal solution of Problem (CMP).
\end{lemma}
\begin{proof}
We prove this by contradiction. Suppose that (\ref{eq_trafficdemand}) is not tight for some MU $i$ at an optimal solution. Then, we can reduce $x_{iB}$ by setting $x_{iB}=R_{i}^{\text{req}}-x_{iA}$, which will reduce the total cost (i.e., the objective function) without violating any constraint. This contradicts the fact that it is an optimal solution, and hence completes the proof.
\end{proof}

By using Lemma \ref{lemma_binding} and eq. (\ref{eq_rate_BS}), we express $x_{iB}$ and $p_{iB}$ as functions of $p_{iA}$ for each MU $i$ as follows:
\begin{eqnarray}
x_{iB} &=& R_{i}^{\text{req}} -W \log_{2} \big(1+\frac{p_{iA} g_{iA}}{\sum_{j\neq i,j\in\mathcal{I}}p_{jA} g_{jA}+n_A}\big),\forall i\in\mathcal{I}, \label{result_xiB} \\
p_{iB} &=& \frac{n_B}{g_{iB}}2^{\frac{R_{i}^{\text{req}}}{B}}
\frac{1}{(1+\frac{p_{iA} g_{iA}}{\sum_{j\neq i,j\in\mathcal{I}}p_{jA} g_{jA}+n_A})^{\frac{W}{B}}}-\frac{n_B}{g_{iB}},\forall i\in\mathcal{I}. \label{result_piB}
\end{eqnarray}

Similar to \cite{Paper:Backhaul_Duan}\cite{Paper:Economic_Huang}, we consider the practical scenario that the price of the AP is lower than that of the BS (i.e., $\pi_A<\pi_B$), which motivates the MUs to offload traffic to the AP. By using (\ref{eq_rate_AP}), (\ref{result_xiB}), and (\ref{result_piB}), we transform Problem (CMP) into an equivalent Transmit-Power Allocation Problem (TPA-P):
\begin{eqnarray}
\text{(TPA-P): } &&\text{Maximize} \sum_{i\in\mathcal{I}} (\pi_B-\pi_A) W \log_{2} \big(1+\frac{p_{iA} g_{iA}}{\sum_{j\neq i,j\in\mathcal{I}}p_{jA}g_{jA}+n_A}\big) \nonumber \\
&& \text{Subject to: }W \log_{2} \big(1+\frac{p_{iA} g_{iA}}{\sum_{j\neq i,j\in\mathcal{I}}p_{jA}g_{jA}+n_A}\big) \leq R_{i}^{\text{req}},\forall i\in\mathcal{I}, \label{constraint_P3_rate} \\
&& \text{~~~~~~~~~~~~~~}p_{iA}\leq  P_{iA}^{\max},\forall i\in\mathcal{I},\label{constraint_P3_PiAmax}\\
&& \text{~~~~~~~~~~~~~~}\frac{n_B}{g_{iB}}2^{\frac{R_{i}^{\text{req}}}{B}}
\frac{1}{(1+\frac{p_{iA}g_{iA}}{\sum_{j\neq i,j\in\mathcal{I}}p_{jA}g_{jA}+n_A})^{\frac{W}{B}}} \leq
P_{iB}^{\max} + \frac{n_B}{g_{iB}}, \forall i\in\mathcal{I}, \label{constraint_P3_PiBmax} \\
&& \text{~~~~~~~~~~~~~~} \frac{n_B}{g_{iB}}2^{\frac{R_{i}^{\text{req}}}{B}}
\frac{1}{(1+\frac{p_{iA}g_{iA}}{\sum_{j\neq i,j\in\mathcal{I}}p_{jA}g_{jA}+n_A})^{\frac{W}{B}}} + p_{iA}\leq P_{i}^{\max} + \frac{n_B}{g_{iB}}, \forall i\in\mathcal{I}, \label{constraint_P3_Pimax} \\
&&  \text{Variables: } p_{iA}, \forall i\in\mathcal{I}. \nonumber
\end{eqnarray}
Problem (TPA-P) only involves $\{p_{iA}\}_{i\in\mathcal{I}}$ as variables. Here, (\ref{constraint_P3_rate}) comes from Lemma \ref{lemma_binding}, i.e., each MU $i$'s $x_{iA}$ cannot exceed $R_{i}^{\text{req}}$. Constraints (\ref{constraint_P3_PiAmax}), (\ref{constraint_P3_PiBmax}), and (\ref{constraint_P3_Pimax}) come from (\ref{P1_Con_PiAmax}), (\ref{P1_Con_PiBmax}), and (\ref{P1_Con_Pimax}), respectively.

Problem (TPA-P) indicates a tradeoff in offloading traffic as follows. To minimize the MUs' total cost, all MUs should offload their traffic to the AP as much as possible. However, aggressive traffic offloading causes a heavy interference at the AP, which increases the MUs' power consumptions (but satisfying (\ref{constraint_P3_PiAmax}), (\ref{constraint_P3_PiBmax}), and (\ref{constraint_P3_Pimax})). Notice that Problem (TPA-P) is more complicated than the well-studied power control problems over interference channels (e.g. \cite{Paper:TransmitPower_Huang}\cite{Paper:TransmitPower_Mung}), because it takes into account each MU's transmit-powers at two different radio-interfaces. This consequently yields the non-convex constraint (\ref{constraint_P3_Pimax}) which couples each MU $i$'s $p_{iA}$ and $p_{iB}$ together. We also recall that all MUs' $\{p_{iA}\}_{i\in\mathcal{I}}$ are also coupled due to the interference at the AP. Hence, Problem (TPA-P) is very difficult to solve.

\subsection{Transforming Problem (TPA-P) into an SINR-Assignment Problem}
To solve Problem (TPA-P), we need to make some further equivalent transformations. We use $\theta_i$ to denote MU's achieved SINR at the AP as follows:
\begin{eqnarray}
\theta_i=\frac{p_{iA}g_{iA}}{\sum_{j\neq i,j\in\mathcal{I}}p_{jA}g_{jA}+n_A},\forall i\in\mathcal{I}. \label{eq_sinr}
\end{eqnarray}
Based on (\ref{eq_sinr}), we have the following result that connects  $\{\theta_i\}_{i\in\mathcal{I}}$ with $\{p_{iA}\}_{i\in\mathcal{I}}$.

\begin{proposition}
\label{proposition_power_theta}
Given any profile of transmit-powers $\{p_{iA}\}_{i\in\mathcal{I}}$ which is feasible for Problem (TPA-P), the corresponding profile of $\{\theta_i\}_{i\in\mathcal{I}}$ given by (\ref{eq_sinr}) ensures that the following result holds:
\begin{eqnarray}
p_{iA}=\frac{n_A}{g_{iA}}\frac{\theta_i}{1+\theta_i} \frac{1}{1-\sum_{i\in\mathcal{I}}\frac{\theta_i}{1+\theta_i}},\forall i\in\mathcal{I},
\label{eq_piA_thetaiA}
\end{eqnarray}
and we always have $\sum_{i\in\mathcal{I}}\frac{\theta_i}{1+\theta_i}<1$.
\end{proposition}

\begin{proof}
Please refer to Appendix I for the details.
\end{proof}

Based on Proposition \ref{proposition_power_theta}, we can use $\{\theta_i\}_{i\in\mathcal{I}}$ to substitute $\{p_{iA}\}_{i\in\mathcal{I}}$ and transform Problem (TPA-P) into an equivalent SINR-assignment problem as follows:
\begin{eqnarray}
\text{(SINR-P):} && \text{Maximize} \sum_{i\in\mathcal{I}}(\pi_B-\pi_A)W \log_2\left(1+\theta_i\right) \nonumber\\
&& \text{Subject to: } 0\leq \theta_i\leq 2^{\frac{R_i^{\text{req}}}{W}}-1,\forall i\in\mathcal{I},\label{P4_Con_throughput}\\
&& \text{~~~~~~~~~~~~~} \frac{n_A}{g_{iA}}\frac{\theta_i}{1+\theta_i}\frac{1}{1-\sum_{i\in\mathcal{I}}\frac{\theta_i}{1+\theta_i}}\leq P_{iA}^{\max},\forall i\in\mathcal{I}, \label{P4_Con_PiAmax}\\
&& \text{~~~~~~~~~~~~~} \frac{n_B}{g_{iB}}2^{\frac{R_i^{\text{req}}}{B}}\frac{1}{(1+\theta_i)^{\frac{W}{B}}} \leq P_{iB}^{\max} + \frac{n_B}{g_{iB}},\forall i\in\mathcal{I}, \label{P4_Con_PiBmax}\\
&& \text{~~~~~~~~~~~~~} \frac{n_A}{g_{iA}}\frac{\theta_i}{1+\theta_i}\frac{1}{1-\sum_{i\in\mathcal{I}}\frac{\theta_i}{1+\theta_i}}
+\frac{n_B}{g_{iB}}2^{\frac{R_i^{\text{req}}}{B}}\frac{1}{(1+\theta_i)^{\frac{W}{B}}}\leq P_i^{\max} +\frac{n_B}{g_{iB}},\forall i\in\mathcal{I}, \label{P4_Con_Pimax}\\
&& \text{~~~~~~~~~~~~~} \sum_{i\in\mathcal{I}}\frac{\theta_i}{\theta_i+1} < 1, \label{P4_Con_SumTheta} \\
&& \text{Variables: } \theta_{i}, \forall i\in\mathcal{I}. \nonumber
\end{eqnarray}
Notice that we introduce constraint (\ref{P4_Con_SumTheta}) in Problem (SINR-P) for the convenience of later discussions, and Proposition \ref{proposition_power_theta} shows that introducing such a constraint does not reduce the feasible set of Problem (TPA-P). Once we solve Problem (SINR-P), then the optimal solution $\{p_{iA}^\ast\}_{i\in\mathcal{I}}$ of Problem (TPA-P) can be computed based on eq. (\ref{eq_piA_thetaiA}).

In the special case where all MUs' traffic demands are small enough (see Fig. \ref{Figure_Offloading_Ratio} in Sec. \ref{section_numerical}), it is optimal for all MUs to offload their entire traffic to the AP for minimizing the cost. The particular structure of Problem (SINR-P) enables us to derive a sufficient condition for this special case to happen.
\begin{proposition}
\label{proposition_special}
\textit{\underline{(Complete-Offloading Situation)}} The optimal solution of Problem (SINR-P) is $\theta_i^\ast=2^{\frac{R_i^{\text{req}}}{W}}-1,\forall i\in\mathcal{I}$ when both of the following two conditions (i.e., (C1) and (C2)) hold:
\begin{eqnarray}
\text{(C1): }\sum_{i\in\mathcal{I}}\frac{1}{2^{\frac{R_i^{\text{req}}}{W}}}>I-1, \text{and }
\text{(C2): }\frac{n_A}{g_{iA}}\left(1-\frac{1}{2^{\frac{R_i^{\text{req}}}{W}}}\right)
\frac{1}{\sum_{i\in\mathcal{I}}\frac{1}{2^{\frac{R_i^{\text{req}}}{W}}}-I+1} \leq \min\{P_i^{\max},P_{iA}^{\max}\}
,\forall i\in\mathcal{I}.\nonumber
\end{eqnarray}
\end{proposition}

\begin{proof}
We first derive $\theta_i^\ast=2^{\frac{R_i^{\text{req}}}{W}}-1$ by setting constraint (\ref{P4_Con_throughput}) tight, which corresponds to that each MU offloads its entire traffic to the AP. We then derive Conditions (C1) and (C2) by ensuring that $\{\theta_i^\ast\}_{i\in\mathcal{I}}$ is compatible with all constraints. We obtain Condition (C1) by substituting $\theta_i^\ast=2^{\frac{R_i^{\text{req}}}{W}}-1$ into (\ref{P4_Con_SumTheta}), and obtain Condition (C2) by substituting $\theta_i^\ast=2^{\frac{R_i^{\text{req}}}{W}}-1$ into (\ref{P4_Con_PiAmax}) and (\ref{P4_Con_Pimax}).
Notice that (\ref{P4_Con_PiBmax}) is always satisfied, since no traffic is sent to the BS when each MU $i$'s $\theta_i^\ast=2^{\frac{R_i^{\text{req}}}{W}}-1$. Therefore, Conditions (C1) and (C2) together guarantee that the optimal solution of Problem (SINR-P) is $\theta_i^\ast=2^{\frac{R_i^{\text{req}}}{W}}-1,\forall i\in\mathcal{I}$.
\end{proof}

\vspace{0.1in}
\subsection{Equivalent Form of the SINR-assignment Problem with $\{\rho_i\}_{i\in\mathcal{I}}$ and $\rho_0$}
\label{subsec_SINRProblem}
Problem (SINR-P), however, is still difficult to solve due its non-convexity. To solve it efficiently, we further introduce a \textit{one-to-one mapping} as follows
\begin{eqnarray}
\rho_i=\frac{\theta_i}{1+\theta_i},    \Longleftrightarrow \theta_i=\frac{\rho_i}{1-\rho_i}, \forall i\in\mathcal{I}.  \label{eq_M_theta}
\end{eqnarray}
Since $\sum_{i\in\mathcal{I}}\rho_i<1$ always holds (i.e., Proposition \ref{proposition_power_theta}), we introduce another positive variable $\rho_0$ such that $\rho_0+\sum_{i\in\mathcal{I}}\rho_i=1$. Such a condition will help simplify the complicated denominators in (\ref{P4_Con_PiAmax}) and (\ref{P4_Con_Pimax}).

By using $\{\rho_i\}_{i\in\mathcal{I}}$ and $\rho_0$, we equivalently transform Problem (SINR-P) into Problem (SINR-M-P) as follows (the letter ``M" means ``Medium", and Appendix II presents the detailed procedures).
\begin{eqnarray}
\text{(SINR-M-P): } && \text{Minimize} \sum_{i\in\mathcal{I}}(\pi_B-\pi_A)W \log_2\left(\rho_0+\sum_{j\neq i,j\in\mathcal{I}}{\rho_j}\right) \label{P5_Objective}\\
&& \text{Subject to: } 0\leq \rho_i\leq 1-\frac{1}{2^{\frac{R_i^{\text{req}}}{W}}}, \forall i\in\mathcal{I},\label{P5_Con_throughput}\\
&& \text{~~~~~~~~~~~~~}  \frac{n_A}{g_{iA}}\frac{\rho_i}{\rho_0}\leq P_{iA}^{\max},\forall i\in\mathcal{I},\label{P5_Con_PiAmax}\\
&& \text{~~~~~~~~~~~~~}  \rho_0+\sum_{j\neq i,j\in\mathcal{I}}\rho_j \leq \left(\frac{P_{iB}^{\max}+\frac{n_B}{g_{iB}}}{\frac{n_B}{g_{iB}} 2^{\frac{R_i^{\text{req}}}{B}}}\right)^{\frac{B}{W}},\forall i\in\mathcal{I},\label{P5_Con_PiBmax}\\
&& \text{~~~~~~~~~~~~~}  \frac{n_B}{g_{iB}} 2^{\frac{R_i^{\text{req}}}{B}} \left(\rho_0+\sum_{j\neq i,j\in\mathcal{I}}{\rho_j}\right)^{\frac{W}{B}}+
\frac{n_A}{g_{iA}}\frac{\rho_i}{\rho_0}\leq P_i^{\max}+\frac{n_B}{g_{iB}},\forall i\in\mathcal{I},\label{P5_Con_Pimax}\\
&& \text{~~~~~~~~~~~~~}  \rho_0+\sum_{i\in\mathcal{I}} \rho_i= 1,\label{P5_Con_M} \\
&& \text{Variables: } \rho_0\text{ and } \rho_{i},\forall i\in\mathcal{I}. \nonumber
\end{eqnarray}
Notice that constraints (\ref{P5_Con_throughput})-(\ref{P5_Con_M}) here correspond to (\ref{P4_Con_throughput})-(\ref{P4_Con_SumTheta}) in Problem (SINR-P), respectively.

\begin{proposition}
\label{proposition_equivalent_M0} Let ($\rho_0^\ast, \{\rho_i^\ast\}_{i\in\mathcal{I}}$) denote an optimal solution of Problem (SINR-M-P). Then, $\theta_i^\ast=\frac{\rho_i^\ast}{1-\rho_i^\ast},\forall i\in\mathcal{I}$ corresponds to an optimal solution of Problem (SINR-P)\footnote{It is worth emphasizing that Proposition \ref{proposition_equivalent_M0} does not hold if Problem (CMP) is infeasible (i.e., Assumption \ref{assumption_feasible} fails to hold).}.
\end{proposition}

\begin{proof}
The key idea is to show that the feasible region of Problem (SINR-P) and that of (SINR-M-P) form a one-to-one mapping, if Problem (SINR-P) is feasible. Please see Appendix II for the details.
\end{proof}

Proposition \ref{proposition_equivalent_M0} allows us to solve Problem (SINR-P) by solving Problem (SINR-M-P). In particular, Problem (SINR-M-P) can be solved through a \textit{two-layered structure}, i.e., we first solve the subproblem in which the value of $\rho_0$ is given, and then we find the best $\rho_0^\ast\in(0,1]$ to minimize the objective function.

\section{Centralized Algorithm to Solve Problem (SINR-M-P)}
\label{section_centralized}
We propose a centralized algorithm to solve Problem (SINR-M-P). Sec. \ref{subsection_layered} shows the two-layered structure. With the layered structure, we design an algorithm to solve a subproblem of Problem (SINR-M-P) in Sec. \ref{subsection_subproblem}, based on which we design an algorithm to solve the top problem in Sec. \ref{subsection_whole}.

\subsection{Two-Layered Structure of Problem (SINR-M-P)}
\label{subsection_layered}
Problem (SINR-M-P) is still a non-convex optimization. Nevertheless, solving Problem (SINR-M-P) can be vertically separated into two steps, namely, \textit{solving a series of subproblems in which the value of $\rho_0$ is given in advance, and then solving a top-problem that finds the best $\rho_0^\ast\in(0,1]$ for minimizing the objective function}. The details are as follows.

\subsubsection{\underline{(Subproblem under a given $\rho_0$)}} We consider a subproblem under a given $\rho_0\in(0,1]$ as follows:
\begin{eqnarray}
\text{(SINR-M-SubP): }  F_{\text{sub}}(\rho_0)=&& \text{Minimize} \sum_{i\in\mathcal{I}}(\pi_B-\pi_A)W \log_2\left(\rho_0+\sum_{j\in\mathcal{I},j\neq i}{\rho_j}\right) \label{eq_obj_SubP} \\
&&  \text{ Subject to: } \{\rho_i\}_{i\in\mathcal{I}} \in \mathcal{G}_{(\rho_0)} \cap \mathcal{H}_{(\rho_0)}, \label{constraint_SINR_M_SubP} \\
&&  \text{ Variables: } \rho_i,\forall i\in\mathcal{I}. \nonumber
\end{eqnarray}
Constraint (\ref{constraint_SINR_M_SubP}) means that the profile $\{\rho_i\}_{i\in\mathcal{I}}$ for all MUs should belong to the intersection of sets $\mathcal{G}_{(\rho_0)}$ and $\mathcal{H}_{(\rho_0)}$. Here, set $\mathcal{G}_{(\rho_0)}$ can be characterized by (\ref{P5_Con_throughput}), (\ref{P5_Con_PiAmax}), (\ref{P5_Con_PiBmax}), and (\ref{P5_Con_Pimax}) under a given $\rho_0$ as follows:
\begin{eqnarray}
&& \mathcal{G}_{(\rho_0)}=\big\{\{\rho_i\}_{i\in\mathcal{I}}| 0\leq \rho_i\leq 1-\frac{1}{2^{\frac{R_i^{\text{req}}}{W}}}, \forall i\in\mathcal{I};\text{ }\frac{n_A}{g_{iA}}\frac{\rho_i}{\rho_0}\leq P_{iA}^{\max},\forall i\in\mathcal{I};
\text{ } \nonumber \\
&& \text{~~~~~~~~~}  \sum_{j\neq i,j\in\mathcal{I}}\rho_j \leq \left(\frac{P_{iB}^{\max}+\frac{n_B}{g_{iB}}}{\frac{n_B}{g_{iB}} 2^{\frac{R_i^{\text{req}}}{B}}}\right)^{\frac{B}{W}}-\rho_0,\forall i\in\mathcal{I}; \text{ } \nonumber \\
&& \text{~~~~~~~~~}  \frac{n_B}{g_{iB}} 2^{\frac{R_i^{\text{req}}}{B}} \left(\rho_0+\sum_{j\neq i,j\in\mathcal{I}}{\rho_j}\right)^{\frac{W}{B}}+
\frac{n_A}{g_{iA}}\frac{\rho_i}{\rho_0}\leq P_i^{\max}+\frac{n_B}{g_{iB}},\forall i\in\mathcal{I}\big\}. \label{eq_normalset}
\end{eqnarray}
Meanwhile, set $\mathcal{H}_{(\rho_0)}$ is characterized by constraint (\ref{P5_Con_M}) under a given $\rho_0$ as follows:
\begin{eqnarray}
\mathcal{H}_{(\rho_0)}=\big\{\{\rho_i\}_{i\in\mathcal{I}}| \sum_{i\in\mathcal{I}} \rho_i\geq 1 - \rho_0 \big\}. \label{eq_reversednormalset}
\end{eqnarray}
Note that we change the equality in (\ref{P5_Con_M}) into the inequality in $\mathcal{H}_{(\rho_0)}$, and this does not change the optimal solution due to Proposition \ref{proposition_P5R}. Let $\{\rho_{i,(\rho_0)}^{\ast,\text{sub}}\}_{i\in\mathcal{I}}$ denote an optimal solution of Problem (SINR-M-SubP).

\begin{proposition}
\label{proposition_P5R}
If Problem (SINR-M-SubP) is feasible under a given value of $\rho_0$, then $\sum_{i\in\mathcal{I}} \rho_{i,(\rho_0)}^{\ast,\text{sub}}= 1-\rho_0$ always holds, i.e., the optimal solution of Problem (SINR-M-SubP) is consistent with (\ref{P5_Con_M}).
\end{proposition}
\begin{proof}
Please refer to Appendix III for the details.
\end{proof}

\subsubsection{\underline{(Top-problem to find $\rho_0^\ast$)}} By solving the bottom problem and obtaining $F_{\text{sub}}(\rho_0)$ as a function of a given $\rho_0$, we next solve the top-problem to find the best $\rho_0^\ast$ for minimizing the objective function:
\begin{eqnarray}
\text{(SINR-M-TopP): } \text{Minimize } F_{\text{sub}}(\rho_0), \text{ Subject to: } 0<\rho_0\leq 1,  \text{~~Variable: } \rho_0. \nonumber
\end{eqnarray}

\subsection{Monotonicity of Problem (SINR-M-SubP) and Algorithm (Cen-Sub) for Solution}
\label{subsection_subproblem}
We focus on solving Problem (SINR-M-SubP) in this subsection.

\subsubsection{\underline{(Monotonicity of Problem (SINR-M-SubP))}} We solve Problem (SINR-M-SubP) by using techniques from \textit{monotonic optimization}, which tackles a class of optimization problems with monotonic objective functions and monotonic constraints \cite{Book:Monotonic}\cite{Paper:HoangTuy}. The feasible region of a monotonic optimization problem can be represented by the intersection of a normal set (i.e., Definition 1 below) and a reversed normal set (Definition 2). Thanks to the monotonic property of the objective function and the constraints, the feasible region can be well approximated by a set of poly-blocks (with an arbitrary desirable precision), and an optimal solution can be attained at one of the vertexes of the poly-blocks. This yields an efficient approach, referred to as the \textit{polyblock-approximation}, to solve the monotonic problems \cite{Paper:HoangTuy}. The advantage of the monotonic optimization is that it does not requires the problem to be convex, hence is widely applicable for solving a wide range of practical engineering problems (see \cite{Book:Monotonic} for more details).

\begin{definition}
\textit{(Normal Set)} A set $\mathcal{G}\subset\mathcal{R}_{+}^{n}$ is normal, if for any two points $x$ and $x'\in\mathcal{R}_{+}^{n}$ with $x'\leq x$ and $x\in\mathcal{G}$, we always have $x'\in\mathcal{G}$.
\end{definition}

\begin{definition}
\textit{(Reversed Normal Set)} A set $\mathcal{H}\subset\mathcal{R}_{+}^{n}$ is a reversed normal set, if for two points $x \text{ and } x'\in\mathcal{R}_{+}^{n}$ with $x'\geq x$ and $x\in\mathcal{H}$, we always have $x'\in\mathcal{H}$.
\end{definition}

With the above terminologies, we can show the following result.
\begin{proposition}
\label{proposition_monotonic}
Given a fixed $\rho_0$, Problem (SINR-M-SubP) is a monotonic optimization problem.
\end{proposition}

\begin{proof}
We first note that the objective function (\ref{eq_obj_SubP}) is increasing in $\{\rho_i\}_{i\in\mathcal{I}}$. Moreover, the left hand sides of (\ref{P5_Con_throughput})-(\ref{P5_Con_Pimax}) are increasing in $\{\rho_i\}_{i\in\mathcal{I}}$, i.e., $\mathcal{G}_{(\rho_0)}$ is a normal set and $\mathcal{H}_{(\rho_0)}$ is a reversed normal set. Hence, Problem (SINR-M-SubP) involves the minimization of an increasing function, subject to a feasible region given by $\mathcal{G}_{(\rho_0)}\cap\mathcal{H}_{(\rho_0)}$. Thus, according to \cite{Book:Monotonic}, (SINR-M-SubP) is a monotonic optimization.
\end{proof}

\subsubsection{\underline{(Algorithm (Cen-Sub) for Solving Problem (SINR-M-SubP))}} Based on Proposition \ref{proposition_monotonic}, we propose Algorithm (Cen-Sub) based on the idea of poly-block approximation to solve Problem (SINR-M-SubP). We should point out that, since Problem (SINR-M-SubP) minimizes an increasing function, Algorithm (Cen-Sub) is designed to construct the series of the poly-blocks that \textit{approximate the lower boundary of the feasible region as close as possible}. This is new in the monotonic optimization literature\footnote{The prior applications of monotonic optimization often focus on maximizing an increasing function \cite{Paper:MonotonicExample2}.}.

\begin{figure}[tbph]
\centering
\includegraphics[scale=0.5]{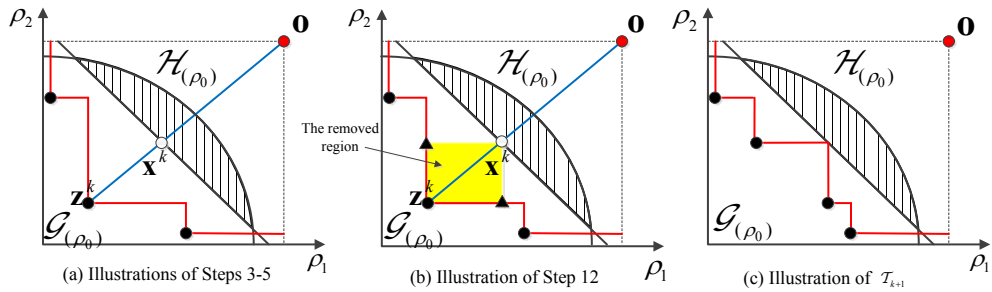}
\vspace{-0.2in}
\caption{The poly-block approximation used in
 Algorithm (Cen-Sub). Each solid node denotes a vertex. The shaded area denotes the feasible region, and the red-line constructed by the vertices denotes the approximated lower-boundary of the feasible region.}
\label{Figure_polyblock}
\end{figure}

\begin{algorithm}
\textbf{Algorithm (Cen-Sub): to solve Problem (SINR-M-SubP) under a given $\rho_0\in(0,1]$}

\vspace{-0.1in}

\hrulefill

\begin{algorithmic}[1]
\STATE \textbf{Initialization:} Set the current best solution $CBS=\emptyset$, and the current best value $CBV=\infty$. Set index $k=1$ and $\epsilon$ as a small positive number. Set the flag for stopping as $f_{\text{stop}}=0$. Initialize set $\mathcal{T}_1=\{\boldsymbol{0}\}$. We use $V(\{\rho_i\}_{i\in\mathcal{I}})=\sum_{i\in\mathcal{I}}(\pi_B-\pi_A)W \log_2\left(\rho_0+\sum_{j\neq i,j\in\mathcal{I}}{\rho_j}\right)$ for easy presentation.
\WHILE{$f_{\text{stop}}=0$}
\STATE Select vertex $\boldsymbol{z}^k\in\arg \min\big\{V(\{\rho_i\}_{i\in\mathcal{I}})|\{\rho_i\}_{i\in\mathcal{I}}\in{\mathcal{T}_k}\big\}$.
\STATE Construct a line between $\boldsymbol{z}^k$ and point $\boldsymbol{o}$ whose element $o_i=\min\big\{1-\frac{1}{2^{\frac{R_i^{\text{req}}}{W}}},\rho_0\frac{P_{iA}^{\max}g_{iA}}{n_A},1-\rho_0\big\},\forall i\in\mathcal{I}$.
\STATE Find the intersection point $\boldsymbol{x}^k$ between the above constructed line and the lower boundary given in $\mathcal{H}_{(\rho_0)}$ by bisection search.

\IF{$V(\boldsymbol{x}^k)<CBV$}
\STATE Update $CBV=V(\boldsymbol{x}^k)$ and set CBS=$\boldsymbol{x}^k$.
\ENDIF
\IF {$\parallel \boldsymbol{x}^k - \boldsymbol{z}^k\parallel<\epsilon$}
\STATE Set $f_{\text{stop}}=1$.
\ENDIF
\STATE Update the set of vertexes as $\mathcal{T}_{k+1}=({\mathcal{T}_k}\backslash\{\boldsymbol{z}^k\})\cup
\big\{\boldsymbol{z}^k+({x}_i^k-{z}_i^k)\boldsymbol{e}_i,i\in\mathcal{I}\big\}$.
\STATE Remove all vertexes $\boldsymbol{z}\in\mathcal{T}_{k+1}\backslash\mathcal{G}_{(\rho_0)}$.
\IF{$\mathcal{T}_{k+1}$ is empty}
\STATE Set $f_{\text{stop}}=1$.
\ENDIF
\STATE Set $k=k+1$.
\ENDWHILE
\STATE \textbf{Output}: Set $\{\rho_{i,(\rho_0)}^{\ast,\text{sub}}\}_{i\in\mathcal{I}}$ is equal to CBS, and $F_{\text{sub}}(\rho_0)=CBV$.
\end{algorithmic}
\end{algorithm}

The key component of Algorithm (Cen-Sub) is the While-Loop (Lines 2-18), whose purpose is to iteratively construct the poly-blocks that approximate the lower-boundary of $\mathcal{G}_{(\rho_0)}\cap\mathcal{H}_{(\rho_0)}$ with an increasing precision. Figure \ref{Figure_polyblock} provides a sketch to illustrate the procedures. Specifically, in the $k$-th iteration, set $\mathcal{T}_k$ denotes the current set of vertexes. In $\mathcal{T}_k$, we find a vertex $\boldsymbol{z}^k$ that yields the smallest objective value (Line 3)\footnote{We use the vector $\boldsymbol{x}$ (in bold letter) to denote the profile $\{x_i\}_{i\in\mathcal{I}}$, i.e., $\boldsymbol{x}$ and $\{x_i\}_{i\in\mathcal{I}}$ are interchangeable.}. Then we perform the following two tasks:
\begin{itemize}
\item \textit{\underline{Task i): to update the current best solution (CBS) and the current best value (CBV)}}. We first construct a line from $\boldsymbol{z}^k$ to a special upper-boundary point $\boldsymbol{o}$ with its element $o_i=\min\big\{1-\frac{1}{2^{\frac{R_i^{\text{req}}}{W}}},\rho_0\frac{P_{iA}^{\max}g_{iA}}{n_A},1-\rho_0\big\},\forall i\in\mathcal{I}$ (Line 4) (each element $o_i$ of point $\boldsymbol{o}$ represents the maximum possible value of $\rho_i$, based on (\ref{P5_Con_throughput}),(\ref{P5_Con_PiAmax}), and (\ref{P5_Con_M})). We then find the corresponding intersection point (denoted by $\boldsymbol{x}^k$) between the constructed line and the lower boundary of the feasible region (Line 5)\footnote{Thanks to the monotonicity of the constraints, we can use the bisection search to find $\boldsymbol{x}^k$ efficiently.}. We use $V(\boldsymbol{x}^k)$ and $\boldsymbol{x}^k$ to update the CBV and the CBS in the $k$-th iteration (Lines 6-8).

\item \textit{\underline{Task ii): to construct poly-blocks $\mathcal{T}_{k+1}$ for the next round iteration}}. We use vertex $\boldsymbol{z}^k$ and the intersection point $\boldsymbol{x}^k$ to construct the new poly-blocks that can approximate $\mathcal{G}_{(\rho_0)}\cap\mathcal{H}_{(\rho_0)}$ closer (Line 12)\footnote{In Line 12, scalar ${x}_i^k$ (or $z_i^k$) denotes the $i$-th element of vector $\boldsymbol{x}^k$ (or vector $\boldsymbol{z}^k$), and vector $\boldsymbol{e}_i$ denotes the vector with the $i$-th element equal to 1, and all other elements equal to 0. All vectors in this paper are of dimension $1\times I$.}. The essence of Line 12 is to remove the region in which the optimum cannot exist.
\end{itemize}

Algorithm (Cen-Sub) terminates if $\boldsymbol{z}^k$ and $\boldsymbol{x}^k$ are close enough (Lines 9-11), or if we cannot expect to find a better solution (Lines 14-16).

\subsection{Algorithm (Cen) for Solving Problem (SINR-M-TopP) and Solution of Problem (CMP)}
\label{subsection_whole}

\subsubsection{\underline{Algorithm (Cen) to solve Problem (SINR-M-TopP)}}

Next we propose Algorithm (Cen) to solve Problem (SINR-M-TopP), by using Algorithm (Cen-Sub) as a subroutine to solve the subproblem. Algorithm (Cen) performs a one-dimensional linear search over $\rho_0\in(0,1]$ with the step-size $\Delta_{\text{top}}$ (the While-Loop on Lines 2-11). For each given value of $\rho_0$, we use Algorithm (Cen-Sub) to evaluate $F_{\text{sub}}(\rho_0)$  (Line 5) and update the currently best solution (Line 7). We perform the feasibility-test for Problem (SINR-M-SubP) under each $\rho_0$ (Line 3), such that we avoid invoking Algorithm (Cen-Sub) when subproblem (SINR-M-SubP) is infeasible and hence save the computational time.

\begin{algorithm}
\textbf{Algorithm (Cen): to solve Problem (SINR-M-TopP)}

\vspace{-0.1in}

\hrulefill

\begin{algorithmic}[1]
\STATE \textbf{Initialization:} Set a small step-size $\Delta_{\text{top}}$. Set $\rho_0=\Delta_{\text{top}}$. Set the CBV as a very large number.
\WHILE{$\rho_0<1$}
\STATE Check the feasibility of Problem (SINR-M-SubP) with Algorithm (Cen-Sub-FC) in Appendix IV.
\IF {Problem (SINR-M-SubP) is feasible}
\STATE Use Algorithm (Cen-Sub) to obtain $\{\rho_{i,(\rho_0)}^{\ast,\text{sub}}\}_{i\in\mathcal{I}}$ and $F_{\text{sub}}(\rho_0)$.
\IF {$F_{\text{sub}}(\rho_0)<\text{CBV}$}
\STATE Set $\text{CBV}=F_{\text{sub}}(\rho_0)$. Set $\rho_0^\ast=\rho_0$, and $\rho_i^\ast=\rho_{i,(\rho_0)}^{\ast,\text{sub}},\forall i\in\mathcal{I}$.
\ENDIF
\ENDIF
\STATE Update $\rho_0=\rho_0+\Delta_{\text{top}}$.
\ENDWHILE
\STATE \textbf{Output:} $\rho_0^\ast$ and $\{\rho_i^{\ast}\}_{i\in\mathcal{I}}$.
\end{algorithmic}
\end{algorithm}

\begin{remark}
The difficulty in solving Problem (SINR-M-TopP) is that we cannot obtain $F_{\text{sub}}(\rho_0)$ in a closed-form expression. Fortunately, in Problem (SINR-M-TopP) the single variable $\rho_0$ is restricted within a fixed interval $(0,1]$ that is independent on the other parameters. This property allows us to use the one-dimensional linear search with a fixed step-size $\Delta_{\text{top}}$ to find $\rho_0^{\ast}$. A smaller step-size $\Delta_{\text{top}}$ will lead to a more accurate solution in the search. In Sec. \ref{section_numerical}, we adopt $\Delta_{\text{top}}=0.005$ (i.e., 200 samples within $(0,1]$) for most of the simulations, since we find that a $\Delta_{\text{top}}$ smaller than 0.005 yields a very limited performance improvement with a significant increase in computational time (see Table \ref{table_monotonic} for details). $\blacksquare$
\end{remark}

We emphasize that although the proposed Algorithm (Cen) provides an effective approach to solve Problem (SINR-M-P), we cannot claim the obtained $\{\rho_i^\ast\}_{i\in\mathcal{I}}$ as the global optimum due to the use of linear search with a fixed step-size $\Delta_{\text{top}}$. Nevertheless, we have the following asymptotical result.
\begin{proposition}
Algorithm (Cen) can yield the asymptotically optimal solution (i.e., $\rho_0^\ast$ and $\{\rho_i^\ast\}_{i\in\mathcal{I}}$) for Problem (SINR-M-P), as $\Delta_{\text{top}}$ approaches to zero.
\end{proposition}
\begin{proof}
Using an extremely small step-size $\Delta_{\text{top}}$ going to zero enables us to enumerate all possible values of $\rho_0\in(0,1]$. For each enumerated $\rho_0$, the subroutine (Cen-Sub) can find the corresponding optimal solution based on the monotonic optimization theory \cite{Book:Monotonic}\cite{Paper:HoangTuy}. Hence, Algorithm (Cen) can yield the asymptotically optimal solution for Problem (SINR-M-P), as $\Delta_{\text{top}}$ approaches to zero.
\end{proof}

However, Algorithm (Cen) with a very small step-size $\Delta_{\text{top}}$ will consume a very long computational time. We show in Sec. \ref{section_numerical} that by using an appropriately chosen small step-size $\Delta_{\text{top}}=0.005$, Algorithm (Cen) is able to yield the close-to-optimum solution while consuming an affordable computational complexity. The detailed complexity of Algorithm (Cen) depends on both the step-size $\Delta_{\text{top}}$ and the complexity of Algorithm (Cen-Sub). However, according to \cite{Book:Monotonic}\cite{Paper:HoangTuy}, it is very difficult to quantify the complexity of Algorithm (Cen-Sub) directly (because it strongly depends on the geometric structure of the feasible region). In Sec. \ref{section_numerical}, we evaluate the complexity of Algorithm (Cen) via numerical examples.

\subsubsection{\underline{Optimal Solution of Problem (CMP)}} Algorithm (Cen) outputs $\rho_0^\ast$ and $\{\rho_i^{\ast}\}_{i\in\mathcal{I}}$ for Problem (SINR-M-P). Based on Proposition \ref{proposition_equivalent_M0}, each MU $i$'s SINR at the AP (i.e., the solution of Problem (SINR-P)) is $\theta_i^\ast=\frac{\rho_i^\ast}{1-\rho_i^\ast},\forall i\in\mathcal{I}$. Correspondingly, the solution of Problem (CMP) can be derived as follows. Based on Proposition \ref{proposition_power_theta} and eq. (\ref{eq_rate_AP}), MU $i$'s transmit-power and transmission rate to the AP are
\begin{eqnarray}
p_{iA}^\ast=\frac{n_A}{g_{iA}}\frac{\rho_i^\ast}{1-\sum_{i\in\mathcal{I}}\rho_i^\ast} \text{ and } x_{iA}^\ast=W\log_{2}\big(\frac{1}{1-\rho_i^\ast}\big), \forall i\in\mathcal{I}, \nonumber
\end{eqnarray}
respectively. Moreover, based on (\ref{result_xiB}) and (\ref{result_piB}), MU $i$'s rate and transmit-power at the BS are
\begin{eqnarray}
p_{iB}^\ast=\frac{n_B}{g_{iB}}\left(2^{\frac{R_{i}^{\text{req}}}{B}}(1-\rho_i^\ast)^{\frac{W}{B}}-1\right) \text{ and } x_{iB}^\ast=R_i^\text{req} +W\log_{2}\left(1-\rho_i^\ast \right), \forall i\in\mathcal{I}, \nonumber
\end{eqnarray}
respectively. Thus, we solve Problem (CMP) completely.

To perform Algorithm (Cen), the BS needs to collect each MU's \textit{private information}, e.g., the channel gains $\{g_{iA},g_{iB}\}$, the transmit-power capacities $\{P_{iA}^{\max},P_{iB}^{\max},P_{i}^{\max}\}$, and the demand $R_{i}^{\text{req}}$. This may not be always feasible in practice, due to privacy concerns and consideration of communications overhead. This motivates us to study whether it is possible to solve Problem (SINR-M-P) in a distributed fashion.

\section{Distributed Algorithm for Solving Problem (SINR-M-P)}
\label{section_distributed}
We first consider implementing Algorithm (Cen) in a distributed manner, which turns out to be challenging due to the centralized nature of the subroutine (Cen-Sub) to solve the monotonic optimization. Hence, we will focus on a practical network scenario as follows where a distributed algorithm is possible.
\begin{assumption}
\label{assumption_bandwidth}
The channel bandwidth of the small-cell AP is no smaller than the allocated bandwidth by the BS, i.e., $W\geq B$.
\end{assumption}

Assumption \ref{assumption_bandwidth} is motivated by the fact that a small-cell AP often provides a larger bandwidth than the cellular BS. For example, IEEE 802.11ac enables a total channel bandwidth up to 160MHz \cite{Paper:BandwidthAP}, which is larger than that of 60MHz in the 3GPP LTE-A standard \cite{WhitePaper:Qualcomm}\cite{WhitePaper:NSN}. We will later show that Assumption \ref{assumption_bandwidth} helps convexify (\ref{P5_Con_Pimax}) such that Problem (SINR-M-P) becomes a concave minimization problem. This involves some additional equivalent problem transformations as follows.

\subsection{Equivalent Form of Problem (SINR-M-P) and Its Two-Layered Structure}
To find a structure suitable for a distributed solution, we re-express Problem (SINR-M-P) as follows.
\begin{eqnarray}
\text{(SINR-ME-P): } && \text{Minimize} \sum_{i\in\mathcal{I}}(\pi_B-\pi_A)W \log_2\left(1-\rho_i\right) \label{P5_E_Objective}
\end{eqnarray}
\begin{eqnarray}
&& \text{Subject to: }  \text{ constraints } (\ref{P5_Con_throughput}), (\ref{P5_Con_PiAmax}), \text{ and } (\ref{P5_Con_M}) \nonumber \\
&& \text{~~~~~~~~~~~~} 1-\left(\frac{P_{iB}^{\max}+\frac{n_B}{g_{iB}}}{\frac{n_B}{g_{iB}} 2^{\frac{R_i^{\text{req}}}{B}}}\right)^{\frac{B}{W}} \leq \rho_i,\forall i\in\mathcal{I},\label{P5_E_Con_PiBmax}\\
&& \text{~~~~~~~~~~~~} \frac{n_B}{g_{iB}} 2^{\frac{R_i^{\text{req}}}{B}} \left(1-\rho_i\right)^{\frac{W}{B}}+
\frac{n_A}{g_{iA}}\frac{\rho_i}{\rho_0}\leq P_i^{\max}+\frac{n_B}{g_{iB}},\forall i\in\mathcal{I},\label{P5_E_Con_Pimax}\\
&& \text{Variables: }  \rho_0 \text{ and } \rho_i,\forall i\in\mathcal{I}. \nonumber
\end{eqnarray}
Based on (\ref{P5_Con_M}), the above (\ref{P5_E_Objective}), (\ref{P5_E_Con_PiBmax}), and (\ref{P5_E_Con_Pimax}) are equivalent to (\ref{P5_Objective}), (\ref{P5_Con_PiBmax}), and (\ref{P5_Con_Pimax}), respectively (that is why we add additional letter ``E" after ``M" to label the problem). Problem (SINR-ME-P) also has a two-layered structure, namely, we can first solve the subproblem with a given $\rho_0$, and then solve the top-problem to determine the best $\rho_0^\ast\in(0,1]$. The subproblem and the top-problem are given as follows.

\subsubsection{\underline{Subproblem under given $\rho_0$}} The subproblem under a given $\rho_0$ can be given as follows.
\begin{eqnarray}
\text{(SINR-ME-SubP): } F_{\text{sub}}(\rho_0)=&& \text{Minimize} \sum_{i\in\mathcal{I}}(\pi_B-\pi_A)W \log_2\left(1-\rho_i\right) \nonumber\\
&& \text{Subject to: }  \text{constraints } (\ref{P5_Con_throughput}),(\ref{P5_Con_PiAmax}),(\ref{P5_Con_M}),(\ref{P5_E_Con_PiBmax}), \text{ and } (\ref{P5_E_Con_Pimax}), \nonumber \\
&& \text{Variables: }  \rho_i,\forall i\in\mathcal{I}. \nonumber
\end{eqnarray}

\subsubsection{\underline{Top problem for finding $\rho_0^\ast$}} Based on the solution of subproblem, we next solve a top-problem:
\begin{eqnarray}
\text{(SINR-ME-TopP): } \text{Minimize } F_{\text{sub}}(\rho_0),  \text{ Subject to: } 0<\rho_0\leq 1,
\text{ Variable: } \rho_0.  \nonumber
\end{eqnarray}
We emphasize that different from solving Problem (SINR-M-P) in a centralized manner as in Section \ref{section_centralized}, we will solve the above Problems (SINR-ME-SubP) and (SINR-ME-TopP) in a distributed manner.

\subsection{Subproblem (SINR-ME-SubP): A Concave Minimization Problem}
\label{subsection_concaveminimization}
To solve Problem (SINR-ME-SubP) in a distributed manner, we first prove the following.
\begin{proposition}
\label{proposition_reversedconvex_M0}
Under Assumption \ref{assumption_bandwidth}, Problem (SINR-ME-SubP) is a concave minimization problem.
\end{proposition}

\begin{proof}
Given $\rho_0$, (\ref{P5_Con_throughput}), (\ref{P5_Con_PiAmax}), (\ref{P5_Con_M}), and   (\ref{P5_E_Con_PiBmax}) are all linear in $\{\rho_i\}_{i\in\mathcal{I}}$. Meanwhile, under Assumption \ref{assumption_bandwidth}, (\ref{P5_E_Con_Pimax}) is convex and thus defines a convex feasible region for each $\rho_i$. Thus, the feasible region of Problem (SINR-ME-SubP) is a convex set. We can further verify that the objective function is a concave function. Thus, Problem (SINR-ME-SubP) involves \textit{the minimization of a concave function over a convex feasible set, i.e., is a concave minimization problem} \cite{Book:ConcaveMinimization}\cite{Paper:ConcaveMinimization}.
\end{proof}

The concave minimization problem is known as a multi-extremal global optimization. For the general concave minimization optimizations, it is very difficult to characterize an optimality condition, and we only know that the optimal solution(s) exist at one or multiple vertices of the feasible region \cite{Book:ConcaveMinimization}.  Motivated by this observation, several numerical schemes (see \cite{Paper:ConcaveMinimization} for more details) have been proposed for enumerating the vertexes of feasible region. However, these schemes suffer from two common drawbacks: the high computational complexity and the centralized implementation. For instance, in \cite{Paper:ConcaveMinimization_DSL}, the authors proposed a parameterized branch-and-bound algorithm to solve the spectrum balancing problem (which was proved as a concave minimization problem). However, the proposed algorithm is centralized without an upper-bound characterization of the computational complexity.

Therefore, to exploit the concave minimization property to solve Problem (SINR-ME-SubP), we need to identify some additional structural properties of the feasible region in our problem. Fortunately, in Problem (SINR-ME-SubP), only constraint (\ref{P5_Con_M}) couples all MUs' decisions. Now let us tentatively ignore (\ref{P5_Con_M}) (which will be taken into account later on). Then, under a given $\rho_0$, (\ref{P5_Con_throughput}),(\ref{P5_Con_PiAmax}),(\ref{P5_E_Con_PiBmax}), and (\ref{P5_E_Con_Pimax}) yield \textit{a decoupled feasible interval for each MU $i$ in the form of $\rho_i\in\left [\underline{M}_{i,(\rho_0)},\overline{M}_{i,(\rho_0)}\right]$}, which is much simpler to deal with. This property significantly simplifies the feasible region. To analytically characterize $\underline{M}_{i,(\rho_0)}$ and $\overline{M}_{i,(\rho_0)}$, we first define the following function:
\begin{eqnarray}
J_i(\rho_i)=\frac{n_B}{g_{iB}} 2^{\frac{R_i^{\text{req}}}{B}} \left(1-\rho_i\right)^{\frac{W}{B}}+
\frac{n_A}{g_{iA}}\frac{\rho_i}{\rho_0},\forall i\in\mathcal{I}, \label{eq_definition_J}
\end{eqnarray}
with $J_i(0)=\frac{n_B}{g_{iB}} 2^{\frac{R_i^{\text{req}}}{B}}$ and $J_i(1)=\frac{n_A}{g_{iA}}\frac{1}{\rho_0}$.
We have the following results regarding function $J_i(\rho_i)$.

\begin{lemma}
\label{lemma_diffcases}
Under Assumption \ref{assumption_bandwidth} and a given $\rho_0$, each MU $i$'s function $J_i(\rho_i)$ is monotonically increasing within $(0,1]$, if \underline{Condition $\text{(C3): } \frac{n_A}{g_{iA}}\frac{1}{\rho_0}> \frac{n_B}{g_{iB}}2^{\frac{R_i^{\text{req}}}{B}}\frac{W}{B},\forall i\in\mathcal{I}$} holds. If Condition (C3) does not hold, then there exists a special point $\chi_{i,(\rho_0)}$ given by
\begin{eqnarray}
\chi_{i,(\rho_0)}=1-\left(\frac{n_A}{n_B}\frac{g_{iB}}{g_{iA}}\frac{1}{\rho_0}\frac{B}{2^{\frac{R_i^{\text{req}}}{B}}W}\right)^{\frac{B}{W-B}}, \label{value_Z}
\end{eqnarray}
such that function $J_i(\rho_i)$ decreases within $\rho_i\in(0,\chi_{i,(\rho_0)}]$ and increases within $\rho_i\in[\chi_{i,(\rho_0)},1]$.
\end{lemma}

\begin{proof}
Please refer to Appendix V for the details.
\end{proof}

Based on Lemma \ref{lemma_diffcases}, we obtain the following results regarding $\underline{M}_{i,(\rho_0)}$ and $\overline{M}_{i,(\rho_0)}$.

\begin{proposition}
\label{proposition_diffcases}
Under Assumption \ref{assumption_bandwidth} and a given $\rho_0$, constraints (\ref{P5_Con_throughput}),(\ref{P5_Con_PiAmax}),(\ref{P5_E_Con_PiBmax}), and (\ref{P5_E_Con_Pimax}) together yield a decoupled feasible interval for each MU $i$ as $\rho_i\in [\underline{M}_{i,(\rho_0)},\overline{M}_{i,(\rho_0)}]$, where
\begin{eqnarray}
\underline{M}_{i,(\rho_0)}=\max\left\{1-\left(\frac{P_{iB}^{\max}+\frac{n_B}{g_{iB}}}{\frac{n_B}{g_{iB}}
2^{\frac{R_i^{\text{req}}}{B}}}\right)^{\frac{B}{W}},0,\underline{\mu}_{i,(\rho_0)}\right\}, \text{ and }
\overline{M}_{i,(\rho_0)}=\min\left\{1-\frac{1}{2^{\frac{R_i^{\text{req}}}{W}}},
P_{iA}^{\max}\frac{g_{iA}}{n_A}\rho_0,1-\rho_0,\overline{\mu}_{i,(\rho_0)}\right\}. \nonumber
\end{eqnarray}
The tuple of $\left(\underline{\mu}_{i,(\rho_0)},\overline{\mu}_{i,(\rho_0)}\right)$ depends on whether Condition (C3) (in Lemma \ref{lemma_diffcases}) holds or not as follows:

\noindent i) When \underline{Condition (C3) holds}, the tuple of  $\left(\underline{\mu}_{i,(\rho_0)},\overline{\mu}_{i,(\rho_0)}\right)$ is given by
\begin{eqnarray}
\left(\underline{\mu}_{i,(\rho_0)},\overline{\mu}_{i,(\rho_0)}\right)=\left\{
      \begin{array}{ll}
         (1,0), & \hbox{if } J_i(0)>P_i^{\max}+\frac{n_B}{g_{iB}};\\
        \big(0,\arg\max\left\{v|J_i(v)=P_i^{\max}+\frac{n_B}{g_{iB}},0\leq v\leq 1\right\}\big), & \hbox{if } J_i(0)\leq P_i^{\max}+\frac{n_B}{g_{iB}}\leq J_i(1); \\
        (0,1), & \hbox{if } J_i(1)< P_i^{\max}+\frac{n_B}{g_{iB}}.
      \end{array}
    \right. \nonumber
\end{eqnarray}

\noindent ii) When \underline{Condition (C3) does not hold}, the tuple of  $\left(\underline{\mu}_{i,(\rho_0)},\overline{\mu}_{i,(\rho_0)}\right)$ is given by
\begin{eqnarray}
\left(\underline{\mu}_{i,(\rho_0)},\overline{\mu}_{i,(\rho_0)}\right)=
\text{~~~~~~~~~~~~~~~~~~~~~~~~~~~~~~~~~~~~}
\text{~~~~~~~~~~~~~~~~~~~~~~~~~~~~~~~~~~~~}
\text{~~~~~~~~~~~~~~~~~~~~~~~~~~~~~~~~~~~~}\nonumber \\
\left\{
  \begin{array}{ll}
   (1,0), \hbox{~~~~~~~~~if } J_i(\chi_{i,(\rho_0)})>P_i^{\max}+\frac{n_B}{g_{iB}};
\\
(0,1),  \hbox{~~~~~~~~~if } J_i(\chi_{i,(\rho_0)})\leq P_i^{\max}+\frac{n_B}{g_{iB}}, \text{ and } J_i(0)\leq P_i^{\max}+\frac{n_B}{g_{iB}}, \text{ and } J_i(1)\leq P_i^{\max}+\frac{n_B}{g_{iB}}; \\
    \big(0,\arg\min\left\{v|\chi_{i,(\rho_0)}\leq v\leq 1, J_i(v)=P_i^{\max}+\frac{n_B}{g_{iB}}\right\}\big),
\\ \hbox{~~~~~~~~~~~~~~~~if }
  J_i(\chi_{i,(\rho_0)})\leq P_i^{\max}+\frac{n_B}{g_{iB}},\text{ and } J_i(0)\leq P_i^{\max}+\frac{n_B}{g_{iB}},\text{ and } J_i(1)\geq P_i^{\max}+\frac{n_B}{g_{iB}}; \\
    \big(\arg\max\left\{v|0\leq v\leq \chi_{i,(\rho_0)},J_i(v)=P_i^{\max}+\frac{n_B}{g_{iB}}\right\},1\big),
\\ \hbox{~~~~~~~~~~~~~~~~if } J_i(\chi_{i,(\rho_0)})\leq P_i^{\max}+\frac{n_B}{g_{iB}},\text{ and }J_i(0)\geq P_i^{\max}+\frac{n_B}{g_{iB}},\text{ and }J_i(1)\leq P_i^{\max}+\frac{n_B}{g_{iB}} \\
    \big(\arg\max\left\{v|0\leq v\leq \chi_{i,(\rho_0)},J_i(v)=P_i^{\max}+\frac{n_B}{g_{iB}}\right\},\arg\min\left\{v|\chi_{i,(\rho_0)}\leq v\leq 1, J_i(v)=P_i^{\max}+\frac{n_B}{g_{iB}}\right\}\big),
\\ \hbox{~~~~~~~~~~~~~~~~if } J_i(\chi_{i,(\rho_0)})\leq P_i^{\max}+\frac{n_B}{g_{iB}},\text{ and }J_i(0)\geq P_i^{\max}+\frac{n_B}{g_{iB}},\text{ and }J_i(1)\geq P_i^{\max}+\frac{n_B}{g_{iB}}.
  \end{array}
\right. \nonumber
\end{eqnarray}
\end{proposition}
Notice that the value of $\chi_{i,(\rho_0)}$ is given in (\ref{value_Z}).

\begin{proof}
Please refer to Appendix VI for the details. Note that in the third case when Condition (C3) does not hold, we need to calculate $\arg\min\left\{v|\chi_{i,(\rho_0)}\leq v\leq 1, J_i(v)=P_i^{\max}+\frac{n_B}{g_{iB}}\right\}$. This can be efficiently computed by a bisection search, as $J_i(\rho_i)$ is increasing for $\rho_i\in[\chi_{i,(\rho_0)},1]$ (Lemma \ref{lemma_diffcases}) in this case. For the other similar cases in Proposition \ref{proposition_diffcases}, the bisection search is also applicable.
\end{proof}

Using Proposition \ref{proposition_diffcases}, we can re-write Problem (SINR-ME-SubP) in the following equivalent form.
\begin{eqnarray}
\text{(SINR-MES-SubP): } F_{\text{sub}}(\rho_0)= &&  \text{Minimize} \sum_{i\in\mathcal{I}}(\pi_B-\pi_A)W \log_2\left(1-\rho_i\right) \nonumber\\
&& \text{Subject to: }  \underline{M}_{i,(\rho_0)} \leq \rho_i\leq \overline{M}_{i,(\rho_0)},\forall i\in\mathcal{I}, \label{DISSub_Interval}\\
&& \text{~~~~~~~~~~~~} \text{constraint }  (\ref{P5_Con_M}), \nonumber \\
&& \text{Variables: }  \rho_i,\forall i\in\mathcal{I}.  \nonumber
\end{eqnarray}
In Problem (SINR-MES-SubP), each MU $i$ can calculate its $\underline{M}_{i,(\rho_0)}$ and $\overline{M}_{i,(\rho_0)}$ using only its private information without communicating with the BS.

\begin{remark}
\label{remark_interpretation}
We can interpret Problem (SINR-MES-SubP) as a resource allocation problem. Specifically, given $\rho_0$, the BS allocates a total budget $1-\rho_0$ of virtual currency to all MUs (i.e., (\ref{P5_Con_M})). Variable $\rho_i$ corresponds to the amount of virtual currency allocated to MU $i$, and each MU $i$ has a \textit{dis-utility function} $(\pi_B-\pi_A)W \log_2\left(1-\rho_i\right)$ (according to the objective function) depending on its $\rho_i$. Each MU $i$'s $\rho_i$ should fall within $[\underline{M}_{i,(\rho_0)},\overline{M}_{i,(\rho_0)}]$ (i.e., (\ref{DISSub_Interval})). The BS's objective is to allocate its budget to the MUs to minimize the total dis-utility. $~~~~~~~~~~~~~~~~~~~~~~~~~~~~~~~~~~~~~~~~~~~~~~~~~~~~~~~~~~~~~~~~~~~~~~~~~~~~~~~~~~~~~~\blacksquare$
\end{remark}

\subsection{Distributed Algorithm (Dis-Sub) for Solving Problem (SINR-MES-SubP)}

The concave minimization property of Problem (SINR-MES-SubP) together with its simple form (followed by Proposition \ref{proposition_diffcases}) lead to \textit{two important guidelines} for the BS to allocate virtual currency to the MUs: \underline{\textit{Guideline-I}}, each MU $i$ prefers to obtaining a larger $\rho_i$ (falling within $[\underline{M}_{i,(\rho_0)},\overline{M}_{i,(\rho_0)}]$), since its dis-utility function is decreasing $\rho_i$; and \underline{\textit{Guideline-II}}, receiving a larger $\rho_i$ can yield a larger \textit{marginal decrease} in MU $i$'s dis-utility, since its dis-utility function is concave. With the two guidelines, we design a distributed algorithm (Dis-Sub) to solve Problem (SINR-MES-SubP), which works as follows.
\begin{itemize}
\item MU $i$ initializes its current $\rho_i^{\text{current}}=\underline{M}_{i,(\rho_0)}$, and reports $\rho_i^{\text{current}}$ to the BS. The BS initializes its effective budget as $\Theta^{\text{effective}}=1-\rho_0-\sum_{i\in\mathcal{I}}\rho_i^{\text{current}}$ and sets $\Theta^{\text{current}}=\Theta^{\text{effective}}$.

\item The BS quantizes $\Theta^{\text{effective}}$ with a small step-size $\Delta_{\text{sub}}$, which yields $\frac{\Theta^{\text{effective}}}{\Delta_{\text{sub}}}$ units of currency. The BS performs an iterative process to allocate these units to the MUs (the While-Loop on Lines 2-8).

\item In each iteration, the BS allocates one unit of currency (in the amount $\Delta_{\text{sub}}$) to an appropriate MU based on the MUs' bids. To compete for this unit, MU $i$ calculates its bid $b_i$ (Line 3) representing \textit{how eager MU $i$ is to get the unit}\footnote{Each MU's bid corresponds to its marginal decrease in the dis-utility when obtaining one more unit of the currency. Here, we assume that MUs will truthfully report their bids, and we will leave the study about the incentive issues in a future study.} and submits the bid to the BS. MU $i$ also sends the BS a signaling $f_i$ indicating whether it can further increase its $\rho_i^{\text{current}}$ or not. After receiving all MUs' $\{b_i,f_i\}_{i\in\mathcal{I}}$, the BS selects the MU (let us say MU $k$) which submits the largest bid with $f_i=1$, and grants MU $k$ the unit (Line 4). Correspondingly, MU $k$ increases its $\rho_k^{\text{current}}$ by $\Delta_{\text{sub}}$ (Line 5).
\end{itemize}

The iteration process will terminate if either the BS uses up its effective budget or the BS finds that every MU's $f_i=0$, i.e., every MU has reached the maximum value of $\rho_i^{\text{current}}$.

In a nutshell, Algorithm (Dis-Sub) provides an efficient approach to determine the best vertex of the feasible region that can minimize all MUs' total dis-utility (under a given $\rho_0$).

\begin{algorithm}
\textbf{Algorithm (Dis-Sub): distributed algorithm to solve Problem (SINR-MES-SubP)}

\vspace{-0.1in}

\hrulefill

\begin{algorithmic}[1]
\STATE \textbf{Initialization:} Each MU $i$ sets its $\rho_i^{\text{current}}=\underline{M}_{i,(\rho_0)}$ and $f_i=1$, and reports $(\rho_i^{\text{current}},f_i)$ to the BS. The BS sets the effective budget of the currency as $\Theta^{\text{effective}}=1-\rho_0-\sum_{i\in\mathcal{I}}\rho_i^{\text{current}}$ and $\Theta^{\text{current}}=\Theta^{\text{effective}}$.
\WHILE{$\Theta^{\text{current}}\geq \Delta_{\text{sub}}$ and $\prod_{i\in\mathcal{I}}(1-f_i)=0$}
\STATE MU $i$ calculates its bid as $b_i=\frac{1}{(1-\rho_i^{\text{current}})\ln2}$ and the signalling $f_i=\mathbf{I}(\rho_i^{\text{current}}+\Delta_{\text{sub}} \leq \overline{M}_{i,(\rho_0)})$, and submits the tuple of $(b_i,f_i)$ to the BS.
\STATE The BS collects each MU's tuple of $(b_i,f_i)$.
\STATE The BS selects MU $k=\arg\max_{i\in\mathcal{I}} b_i \times f_i$, and notices MU $k$ for allocating $\Delta_{\text{sub}}$.
\STATE MU $k$, which is noticed by the BS, updates its $\rho_i^{\text{current}}=\rho_i^{\text{current}}+\Delta_{\text{sub}}$.
\STATE The BS updates its budget as $\Theta^{\text{current}}=\Theta^{\text{current}}-\Delta_{\text{sub}}$.
\ENDWHILE
\STATE \textbf{Output:} The BS sets $F_{\text{sub}}(\rho_{0})=(\pi_B-\pi_A)W\sum_{i\in\mathcal{I}}\log_{2}(\frac{1}{b_i\ln2})$. Each MU $i$ sets $\rho_{i,(\rho_0)}^{\ast,\text{sub}}=\rho_i^{\text{current}}$.

\end{algorithmic}
\end{algorithm}

\subsection{Distributed Algorithm (Dis) for Solving Problem (SINR-ME-TopP)}
We next propose a distributed Algorithm (Dis) to solve Problem (SINR-ME-TopP). In Algorithm (Dis), the BS performs a one-dimensional linear search over $\rho_0\in(0,1]$. For each enumerated $\rho_0$, the BS and the MUs invoke Algorithm (Dis-Sub) as a subroutine. Algorithm (Dis) works as follows.
\begin{itemize}
\item In Line 3, each MU $i$ determines its feasible interval $\rho_i\in[\underline{M}_{i,(\rho_0)},\overline{M}_{i,(\rho_0)}]$ based on Proposition \ref{proposition_diffcases}.

\item In Line 5, the BS determines the feasibility of Subproblem (SINR-MES-SubP).

\item If Subproblem (SINR-MES-SubP) is feasible, the BS and the MUs perform Algorithm (Dis-Sub) to evaluate $F_{\text{sub}}(\rho_0)$, based on which the BS and the MUs update their respective currently best solutions (Lines 7-9). Otherwise, the BS continues to evaluate the next value of $\rho_0$ (Line 11).
\end{itemize}

\begin{algorithm}
\textbf{Algorithm (Dis): to solve Problem (SINR-ME-TopP)}

\vspace{-0.1in}

\hrulefill

\begin{algorithmic}[1]
\STATE \textbf{Initialization:} The BS sets a small step-size $\Delta_{\text{top}}$ for updating $\rho_0$, and initializes $\rho_0=\Delta_{\text{top}}$. The BS sets the current best value (CBV) as a very large number.
\WHILE{$\rho_0\leq 1$}
\STATE Given $\rho_0$, each MU $i$ uses Proposition \ref{proposition_diffcases} to determine its $\underline{M}_{i,(\rho_0)}$ and $\overline{M}_{i,(\rho_0)}$. Each MU $i$ reports its individual feasible interval $[\underline{M}_{i,(\rho_0)},\overline{M}_{i,(\rho_0)}]$ to the BS.
\STATE The BS collects each MU $i$'s tuple of $(\underline{M}_{i,(\rho_0)},\overline{M}_{i,(\rho_0)})$.
\IF {$\underline{M}_{i,(\rho_0)}\leq \overline{M}_{i,(\rho_0)},\forall i\in\mathcal{I}$ and $\sum_{i\in\mathcal{I}}\underline{M}_{i,(\rho_0)}\leq 1-\rho_0 \leq \sum_{i\in\mathcal{I}} \overline{M}_{i,(\rho_0)}$ are met}
\STATE The BS and the MUs perform Algorithm (Dis-Sub). As a result, the BS obtains $F_{\text{sub}}(\rho_{0})$, and each MU $i$ obtains $\rho_{i,(\rho_0)}^{\ast,\text{sub}}$.
\IF{$F_{\text{sub}}(\rho_0)<CBV$}
\STATE The BS sets $CBV=F_{\text{sub}}(\rho_0)$, and each MU $i$ sets its $\rho_{i}^\ast=\rho_{i,(\rho_0)}^{\ast,\text{sub}}$.
\ENDIF
\ENDIF
\STATE The BS updates $\rho_0=\rho_0+\Delta_{\text{top}}$.
\ENDWHILE
\STATE \textbf{Output:} Each MU $i$ outputs $\rho_{i}^\ast$ as the solution of Problem (SINR-ME-TopP).

\end{algorithmic}
\end{algorithm}

Line 5 checks the feasibility of Problem (SINR-MES-SubP) based on (27) and (37), and avoids invoking Algorithm (Dis-Sub) when Problem (SINR-MES-SubP) is infeasible, thus saves the computational time. Figure \ref{Figure_ComparisonA1A2} in Sec. \ref{section_numerical} will verify the gain of this operation.

Considering that Algorithm (Dis-Sub) requires no more than $\frac{1-\rho_0}{\Delta_{\text{sub}}}$ iterations, the overall number of iterations required by Algorithm (Dis) will be no more than $\frac{1}{2}\frac{1}{\Delta_{\text{sub}}}\frac{1}{\Delta_{\text{top}}}$.

As a comparison, we emphasize that the centralized Algorithm (Cen) proposed in Sec. \ref{section_centralized} does not require Assumption \ref{assumption_bandwidth} and can be applied to more general settings than the distributed Algorithm (Dis). However, the downside of Algorithm (Cen) is that it only exploits the monotonic structure of Problem (SINR-M-P), and thus requires a longer computational time than Algorithm (Dis). The numerical results in Table \ref{table_algorithmcompare_1} and Fig. \ref{Figure_ComparisonA1A2} will verify this point.

\subsection{Distributed Power Control Algorithm for Achieving $\{\theta_i^\ast\}_{i\in\mathcal{I}}$}
After obtaining $\rho_{i}^\ast$, each MU $i$ can derive the solution for Problem (CMP). MU $i$'s SINR at the AP is $\theta_i^\ast=\frac{\rho_i^\ast}{1-\rho_i^\ast}$ (according to (\ref{eq_M_theta})), and the transmission rate to the AP is $x_{iA}^\ast=W\log_{2}(\frac{1}{1-\rho_i^\ast})$ (according to (\ref{eq_rate_AP})). Meanwhile, MU $i$'s transmission rate to the BS is $x_{iB}^\ast=R_{i}^{\text{req}}+W\log_2(1-\rho_i^\ast)$ (according to (\ref{result_xiB})), and the transmit-power to the BS is $p_{iB}^\ast=\frac{n_B}{g_{iB}}\left(2^{\frac{R_{i}^{\text{req}}}{B}}(1-\rho_i^\ast)^{\frac{W}{B}}-1\right)$ (according to (\ref{result_piB})).

However, as shown in (\ref{eq_piA_thetaiA}), all MUs' $\{p_{iA}^\ast\}_{i\in\mathcal{I}}$ are coupled together. Fortunately, by setting $\theta_i^\ast$ as MU $i$'s \textit{targeted SINR}, the MUs can adopt the \textit{minimum power control scheme} \cite{Paper:PowerConvergence} to reach $\{p_{iA}^\ast\}_{i\in\mathcal{I}}$ in a distributed manner. In the minimum power control scheme, each MU initializes its $p_{iA}=0$ and then keeps on updating $p_{iA}$ based on its measured SINR for achieving $\theta_i^\ast$ targeted. Due to the space limitation, we skip the detailed description of this power control scheme, and interested readers please refer to \cite{Paper:PowerConvergence} for the details. We emphasize that, as the solution for Problem (SINR-P), $\{\theta_i^\ast\}_{i\in\mathcal{I}}$ is guaranteed to be feasible to all MUs' transmit-power constraints, which guarantees that the resulting $\{p_{iA}^\ast\}_{i\in\mathcal{I}}$ (obtained by the minimum power control scheme) is also feasible to the original Problem (CMP).

\section{Numerical Results}
\label{section_numerical}
We consider a network scenario that the BS is located at the origin, and the small-cell AP is located at $(350\text{m},0\text{m})$. The MUs are randomly located within a circle, whose center is $(320\text{m},0\text{m})$ and the radius is $20$m. This means that the MUs are closer to the AP than to the BS (otherwise, there is little benefit of considering traffic offloading). We use the similar method as \cite{Paper:Channel} to model the channel power gain, i.e., $g_{iA}=\frac{\varrho_{iA}}{l_{iA}^\kappa}$, where $l_{iA}$ denotes the distance between MU $i$ and the AP, and $\kappa$ denotes the power-scaling factor for the path-loss (we set $\kappa=4$). We further assume that $\varrho_{iA}$ follows an exponential distribution with unit mean due to channel fading.  Figure \ref{Figure_topologys} plots two examples of the network scenarios (namely, an 8-MU scenario and a 12-MU scenario), which will be used in the following simulations.

\begin{figure}[tbph]
\centering
\subfigure[8-MU scenario: 8 MUs are randomly located]{
\label{Fig.8MUs}
\includegraphics[width=7cm,height=6cm]{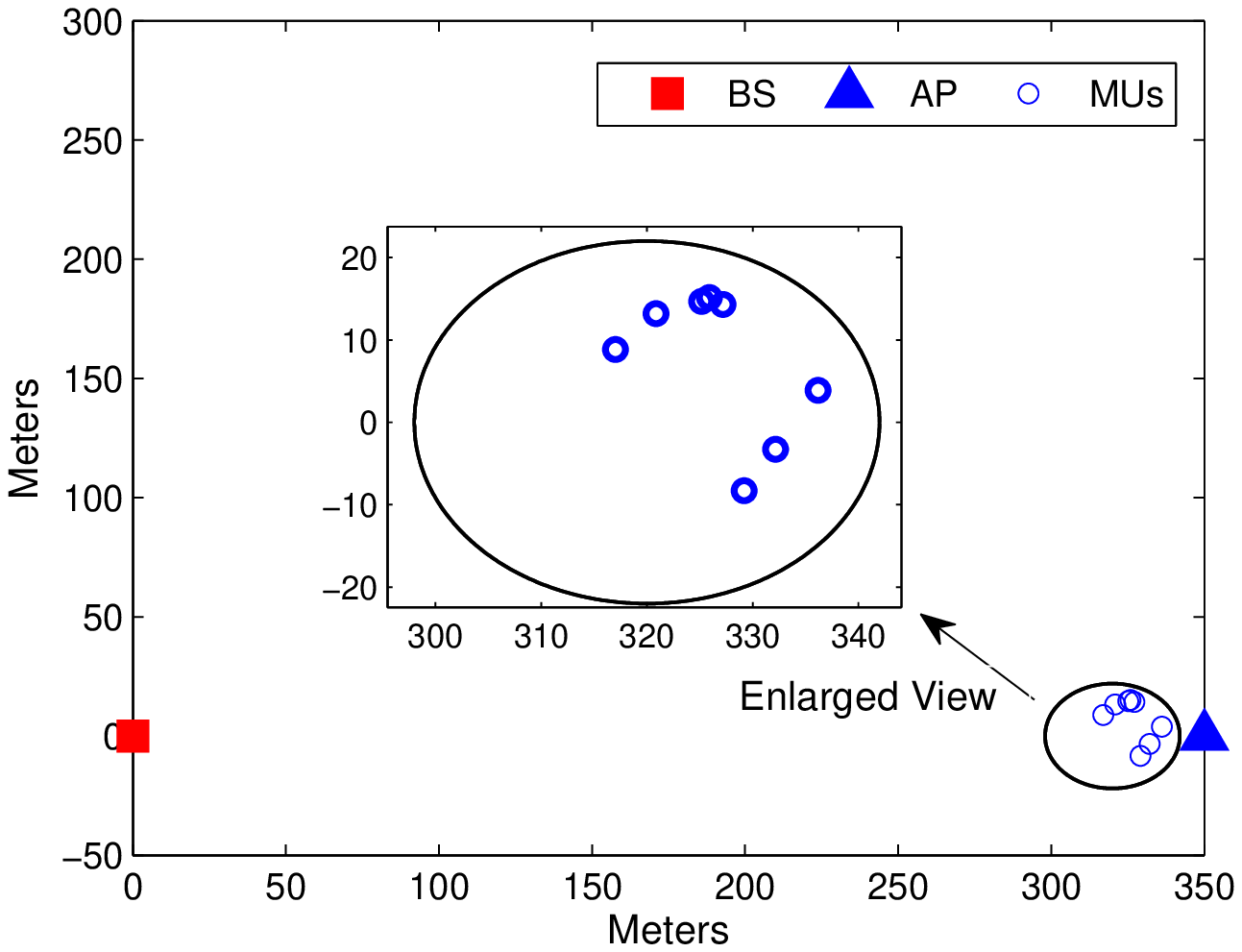}}
\subfigure[12-MU scenario: 12 MUs are randomly located]{
\label{Fig.12MUs}
\includegraphics[width=7cm,height=6cm]{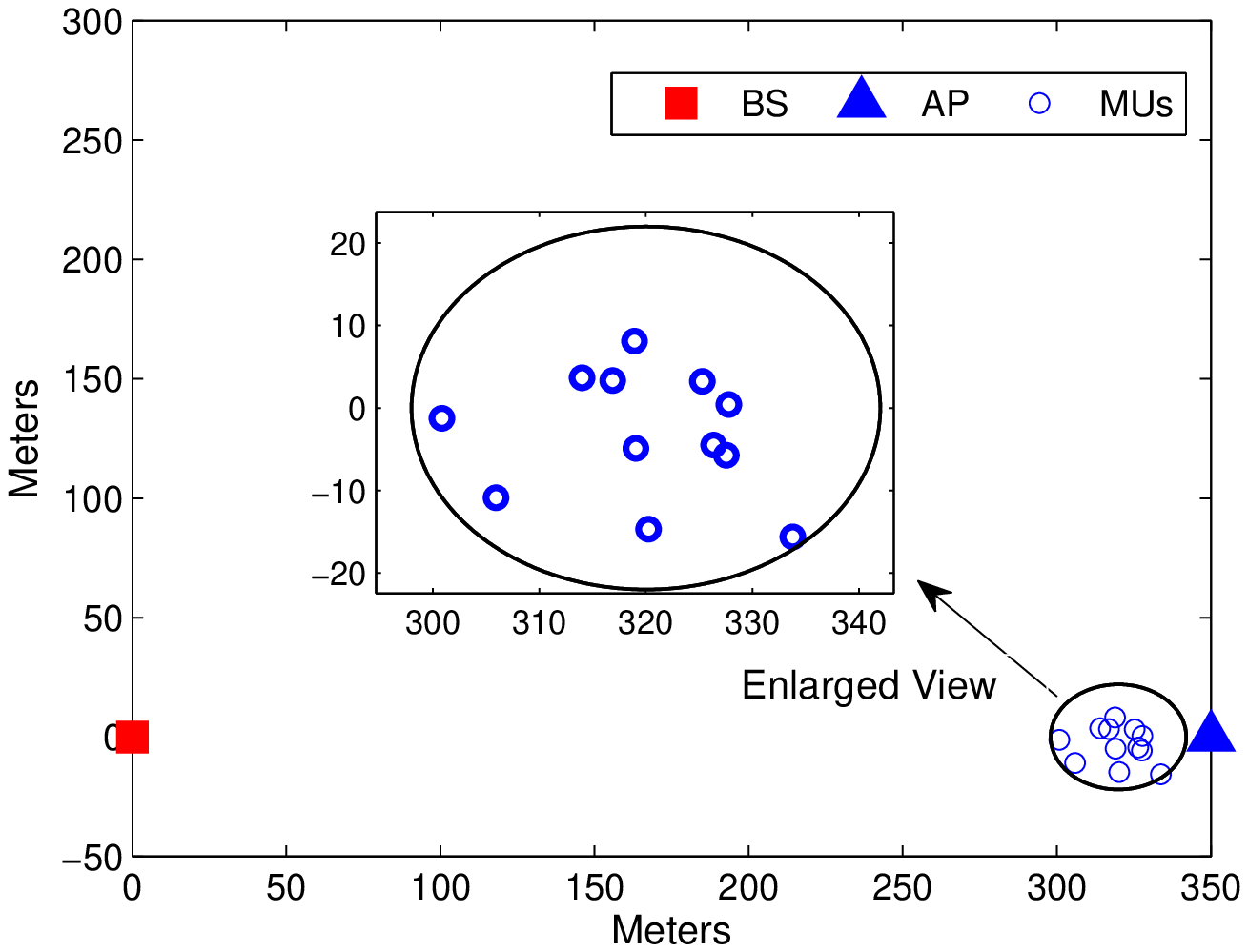}}
\caption{Illustration of network topologies. For the sake of clear presentation, we plot the MUs' positions in an enlarged view.}
\label{Figure_topologys}
\end{figure}

We set the total channel bandwidth of the AP shared by all MUs as $W=20$MHz (802.11a/b/g/n standard \cite{Paper:BandwidthAP}) and the channel bandwidth of the BS to a single MU as $B=5$MHz (close to a WCDMA channel \cite{Paper:PowerNI}). For each MU, we set $P_{iB}^{\max}=0.25$W (i.e., Power Class-3 of mobile devices) , $P_{iA}^{\max}=0.2$W \cite{Paper:PowerNI}, $P_i^{\max}=0.35$W, and $n_0=1\times 10^{-15}$W/Hz. To account for the economic cost, we set $\pi_{B}=\$10/$GB \cite{Paper:Price} and $\pi_{A}=\$2/$GB. As stated earlier, we use $\Delta_{\text{top}}=0.005$ in Algorithms (Cen) and (Dis), and we validate the choice of $\Delta_{\text{top}}=0.005$ at the end of this section by comparing different $\Delta_{\text{top}}$.

\subsection{Accuracy and Efficiency of Algorithm (Cen) and Algorithm (Dis)}

\begin{table}[htbp]
        \caption{Comparison between Algorithm (Cen) and Algorithm (Dis): the 8-MU scenario
        \vspace{-0.4in}}
        \begin{center}
        \begin{tabular}{|c||c|c|c|c|c|c|c|}
  \hline
&$R_i^{\text{req}}=2\text{Mbps}$&$R_i^{\text{req}}=3\text{Mbps}$&$R_i^{\text{req}}=4\text{Mbps}$&$R_i^{\text{req}}=5\text{Mbps}$&$R_i^{\text{req}}=6\text{Mbps}$&$R_i^{\text{req}}=7\text{Mbps}$&$R_i^{\text{req}}=8\text{Mbps}$\\
 \hline
LINGO&0.032,\text{~}3649s&0.048,\text{~}3626s&0.073,\text{~}3668s&0.149,\text{~}3579s&0.225,\text{~}3614s&0.303,\text{~}3643s&0.381,\text{~}3629s\\
\hline
Cen&0.033,\text{~}98.1s&0.049,\text{~}120.2s&0.074,\text{~}123.1s&0.152,\text{~}75.7s&0.232,\text{~}71.6s&0.312,\text{~}68.8s&0.393,\text{~}54.7s\\
 \hline
Dis&0.032,\text{~}2.91s&0.049,\text{~}3.26s&0.074,\text{~}3.57s&0.150,\text{~}3.35s&0.226,\text{~}3.03s&0.306,\text{~}2.69s&0.392,\text{~}2.51s\\
\hline
\end{tabular}
\end{center}
\vspace{-0.2in}
\label{table_algorithmcompare_1}
\end{table}

Table \ref{table_algorithmcompare_1} presents the accuracy and computational efficiency of Algorithm (Cen) and Algorithm (Dis). We consider an 8-MU scenario with the randomly generated channel gains $\{g_{iA}\}_{i\in\mathcal{I}}=[0.1256,2.8108,0.2201,\\0.0381, 0.5091,0.2528,1.4989,0.6081]\times 10^{-4}$ and $\{g_{iB}\}_{i\in\mathcal{I}}=[2.5279.0.6211,1.2604,0.5815,2.5812,\\1.1777,2.6028,2.3551]\times 10^{-8}$. We vary each MU's $R_i^{\text{req}}$ from 2Mbps to 8Mbps (as Problem (CMP) is infeasible when $R_i^{\text{req}}=9$Mbps). For each tested case, the first number represents the minimum total cost and the second number represents the computational time (measured in seconds and obtained by a PC with Intel Core i7-4610M CPU$@$3.00GHz and 8.00GB RAM). To validate Algorithms (Cen) and (Dis), we adopt the global-solver of LINGO to solve Problem (CMP) and obtain the minimum cost. Due to the non-convexity of Problem (CMP), it often takes LINGO a very long time to get an result.

The results in Table \ref{table_algorithmcompare_1} show that Algorithms (Cen) and (Dis) can achieve the results very close to LINGO (with the average relative error equal to $2.89\%$), while both algorithms consume a significantly less computational time than LINGO, because they exploit the special structural properties of the corresponding problems.

\begin{table}[htbp]
        \caption{Comparison between Algorithm (Cen) and LINGO when $W<B$ ($W$=4MHz,$B=5$MHz): the 8-MU scenario}
        \vspace{-0.4in}
        \begin{center}
        \begin{tabular}{|c||c|c|c|c|c|c|c|}
  \hline
&$R_i^{\text{req}}=1\text{Mbps}$&$R_i^{\text{req}}=1.5\text{Mbps}$&$R_i^{\text{req}}=2\text{Mbps}$&$R_i^{\text{req}}=2.5\text{Mbps}$&$R_i^{\text{req}}=3\text{Mbps}$&$R_i^{\text{req}}=3.5\text{Mbps}$&$R_i^{\text{req}}=4\text{Mbps}$\\
 \hline
LINGO&0.03,\text{~}3615s&0.069,\text{~}3609s&0.107,\text{~}3652s&0.146,\text{~}3587s&0.184,\text{~}725s&0.225,\text{~}233s&0.276,\text{~}48s\\
\hline
Cen&0.031,\text{~}54.1s&0.07,\text{~}39.7s&0.11,\text{~}34.5s&0.15,\text{~}28.2s&0.189,\text{~}22.3s&0.228,\text{~}15.7s&0.267,\text{~}7.7s\\
\hline
\end{tabular}
\end{center}
\vspace{-0.2in}
\label{table_algorithmcompare_2}
\end{table}

In Table \ref{table_algorithmcompare_2}, we compare the performance of Algorithm (Cen) and LINGO when $W<B$. The simulation setup is the same as Table I, except that we change $W=4$MHz. In this case, the average relative error of Algorithm (Cen) (compared to LINGO) is $0.67\%$. Notice that Algorithm (Dis) is not applicable in this case because Assumption \ref{assumption_bandwidth} does not hold.

Table \ref{table_algorithmcompare_1} shows that Algorithm (Dis) requires an even shorter time than Algorithm (Cen), because Algorithm (Dis) avoids the complicated procedures for poly-block approximation but exploiting the concave minimization property. Figure \ref{Figure_ComparisonA1A2} presents a detailed comparison of the computational time used by Algorithms (Cen) and (Dis). We use the above 8-MU scenario and another 12-MU scenario in Figure \ref{Figure_ComparisonA1A2}. For the 12-MU scenario, the random channel gains are $\{g_{iA}\}_{i\in\mathcal{I}}=[0.5126,0.8072,2.3568,5.8244,0.6845, 7.1370,\\1.7175,0.4732, 6.2617, 6.7279, 6.0522,0.4312] \times 10^{-5}$, and
$\{g_{iB}\}_{i\in\mathcal{I}}=[2.4780,2.7843,0.5868, 2.7794,\\2.5727,2.4645, 0.6994,0.8504, 2.0981,0.6946,2.6106, 2.9921]\times 10^{-8}$ (for the 12-MU case, Problem (CMP) is infeasible when $R_i^{\text{req}}=8$Mbps).

The top and middle subplots in Fig. \ref{Figure_ComparisonA1A2} show that the computational times used by Algorithms (Cen) and (Dis) first increase as $R_i^{\text{req}}$ increases, and then decrease as $R_i^{\text{req}}$ further increases. To gain a deep understanding of this trend, the bottom subplot plots the \textit{feasible-ratio} when enumerating $\rho_0\in(0,1]$. The feasible-ratio is given by the number of feasible enumerated $\rho_0$ (\textit{which yields a feasible subproblem}) divided by 200 (because of $\Delta_{\text{top}}=0.005$). The bottom subplot shows that the feasible-ratio also first increases as $R_i^{\text{req}}$ increases and then decreases as $R_i^{\text{req}}$ increases in further, which is consistent with the trends in the top and middle subplots. This means that the computational complexity of both Algorithms (Cen) and (Dis) are closely related to the number of enumerated $\rho_0$ which yields a feasible problem. It also demonstrates the effectiveness of removing those infeasible $\rho_0$ from the computation.

\vspace{-0.2in}
\begin{figure}[htbp]
\begin{minipage}[thbp]{0.5\textwidth}
\includegraphics[scale=0.5]{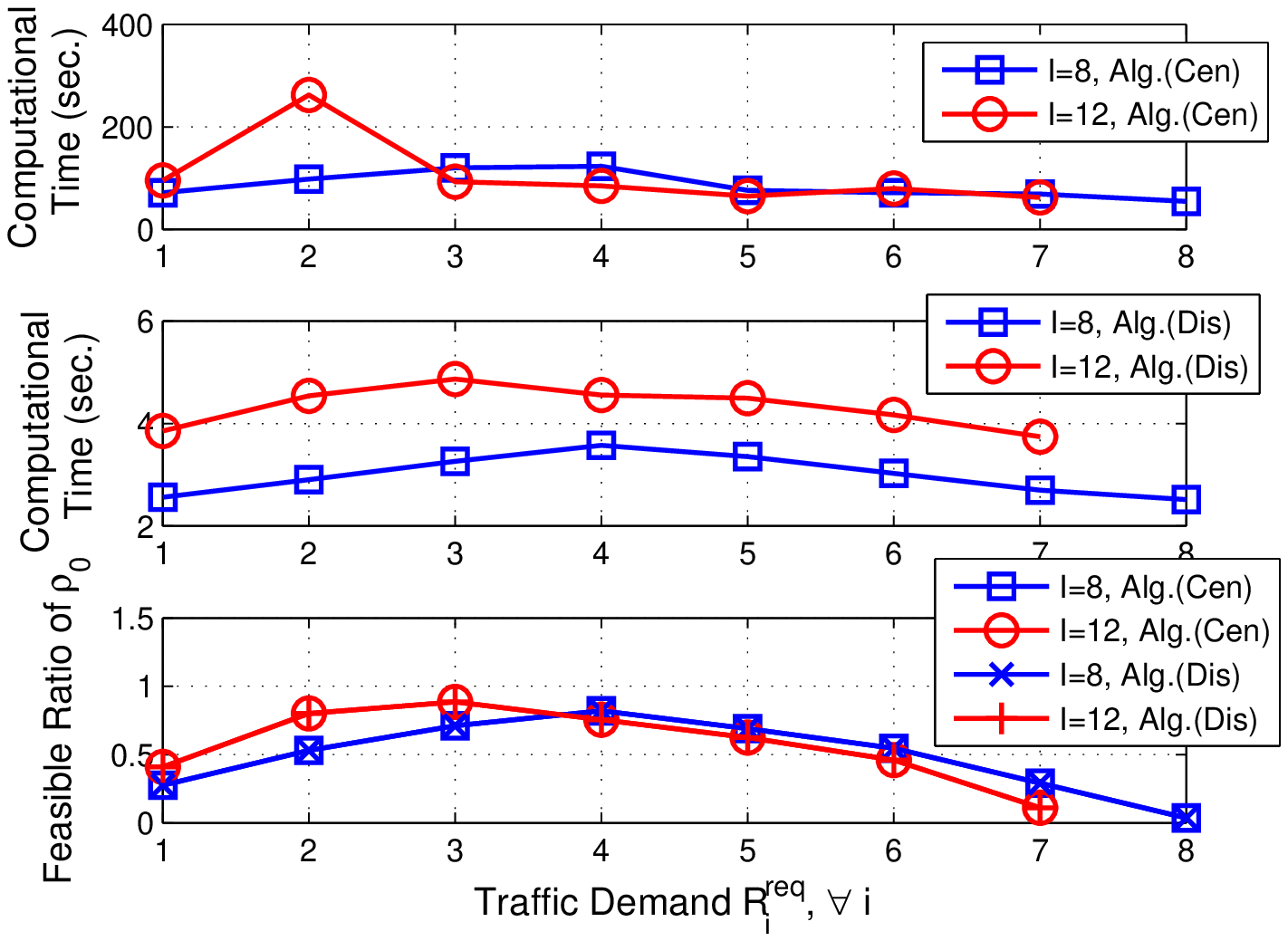}
\vspace{-0.2in}
\caption{Comparison of Computational Time. Top: Time \text{~~~~} used by Alg. (Cen); Middle: Time used by Alg. (Dis); \text{~~~~~~~} Bottom: Feasible-ratio when enumerating $\rho_0\in(0,1]$.}\label{Figure_ComparisonA1A2}
\end{minipage}
\begin{minipage}[thbp]{0.5\textwidth}
\includegraphics[scale=0.5]{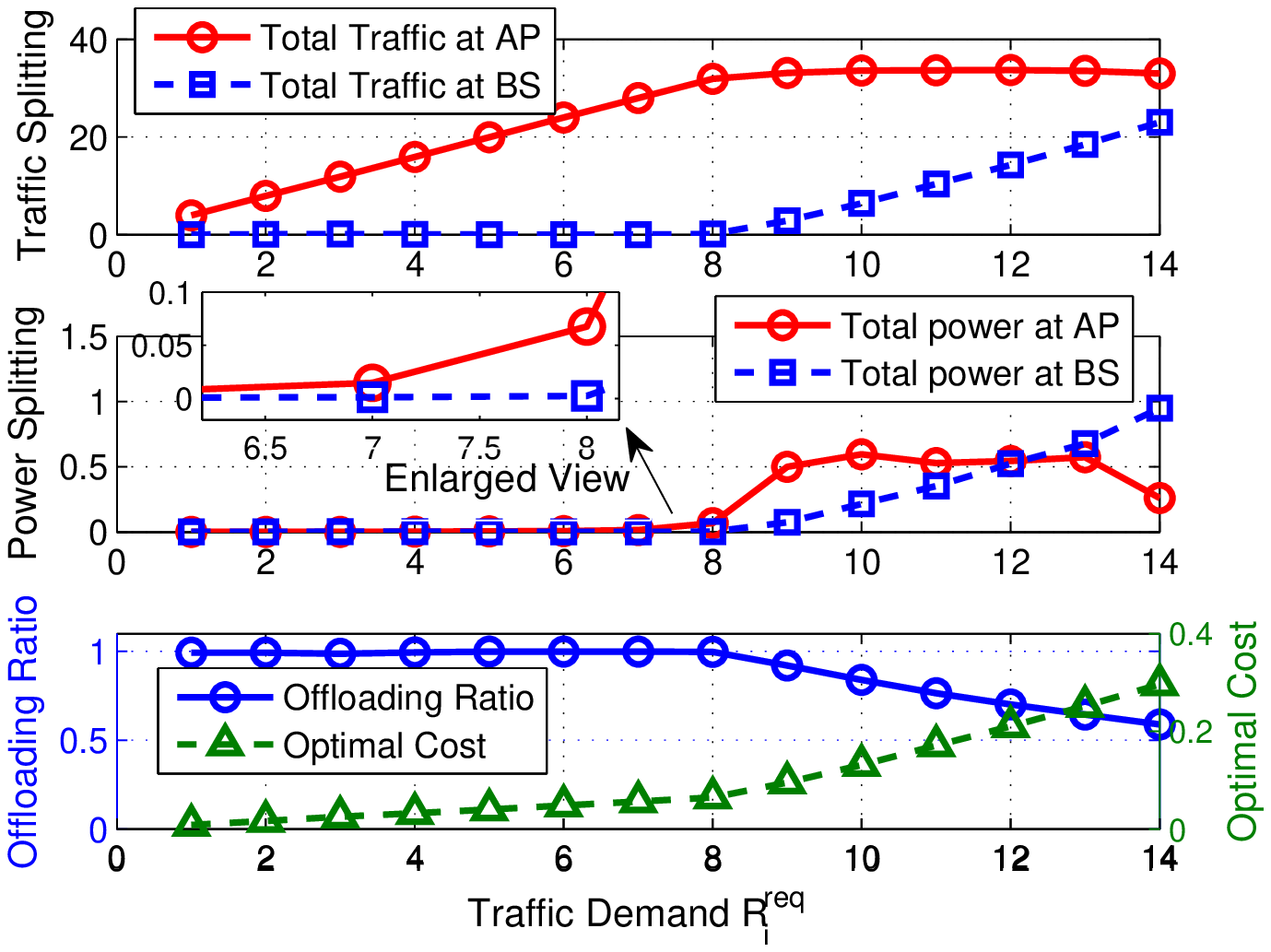}
\vspace{-0.2in}
\caption{Example of Offloading Process for a 4-MU scenario. Top: Traffic to the AP and the BS. Middle: Transmit-power to the AP and the BS. Bottom: Offloading ratio and the total cost.}\label{Figure_Offloading_Ratio}
\end{minipage}
\end{figure}

\subsection{Illustration of Offloading Process}
Figure \ref{Figure_Offloading_Ratio} shows a detailed offloading process as $R_{i}^{\text{req}}$ increases. We consider a 4-MU scenario, in which the random channel power gains are $\{g_{iA}\}_{i\in\mathcal{I}}=[1.2709,0.6407,0.7771,0.8638]\times 10^{-5}$ and $\{g_{iB}\}_{i\in\mathcal{I}}=[3.3164,2.8765,1.4029,2.7934]\times 10^{-8}$. We vary each MU's $R_{i}^{\text{req}}$ from $1$Mbps to $14$Mbps (as Problem (CMP) is infeasible when $R_i^{\text{req}}=15$Mbps). In Fig. \ref{Figure_Offloading_Ratio}, the top subplot shows the traffic to the AP and the BS. The middle subplot shows the corresponding total transmit-powers to the AP and the BS, and the bottom subplot shows the offloading ratio (i.e., the traffic delivered to the AP over the total demand) and the consequent total cost.

\begin{figure}[tbph]
\centering
\subfigure[Each MU's Traffic Scheduling to the AP and BS. \text{~~} Top: Traffic to the AP. Bottom: Traffic to the BS.]{
\label{Figure_Offloading_MUThroughput}
\includegraphics[scale=0.5]{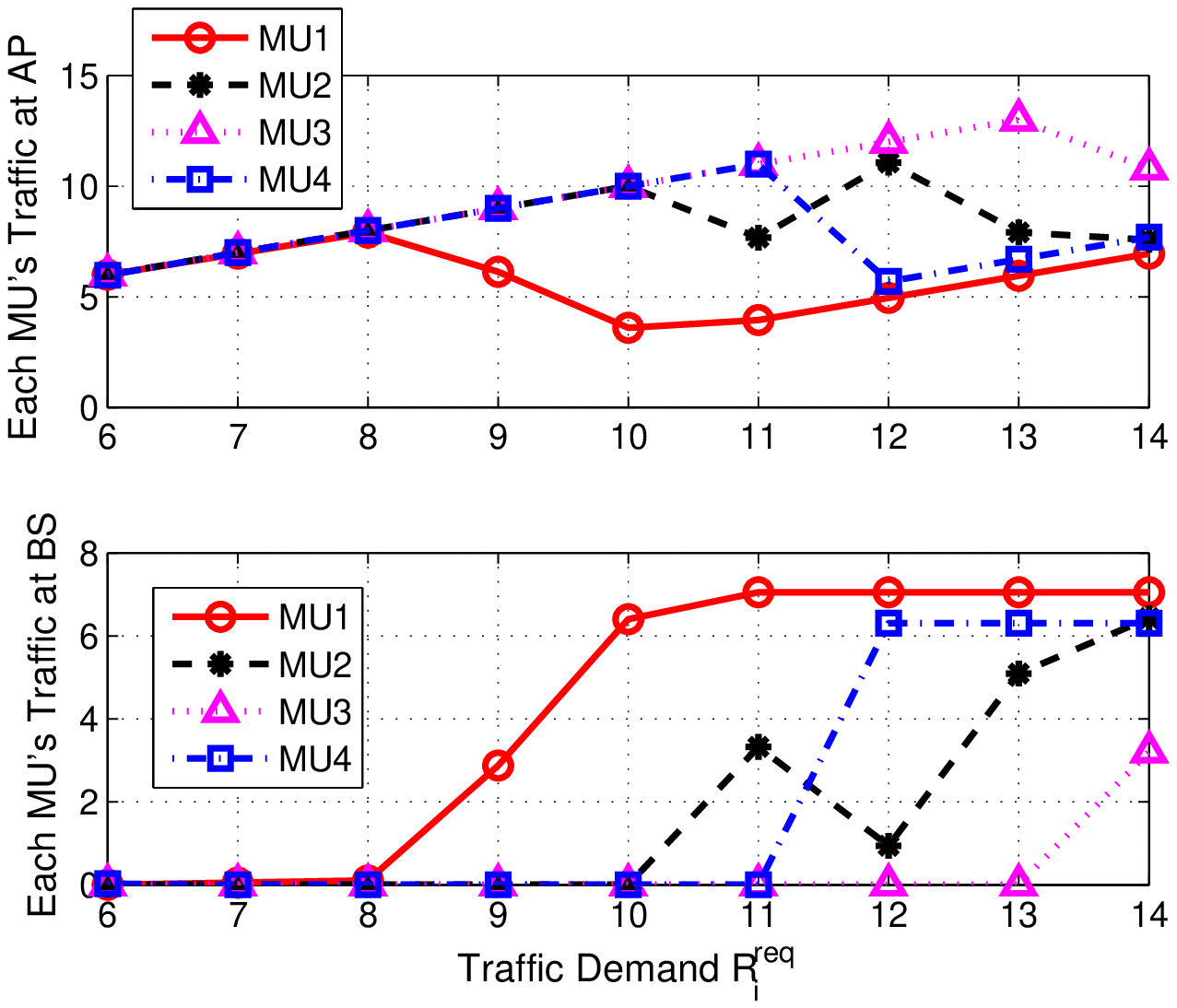}}
\subfigure[Power Allocation. Top: Transmit-Power to the AP. Middle: Transmit-Power to the BS. Bottom: Total Power.]{
\label{Figure_Offloading_MUPower}
\includegraphics[scale=0.5]{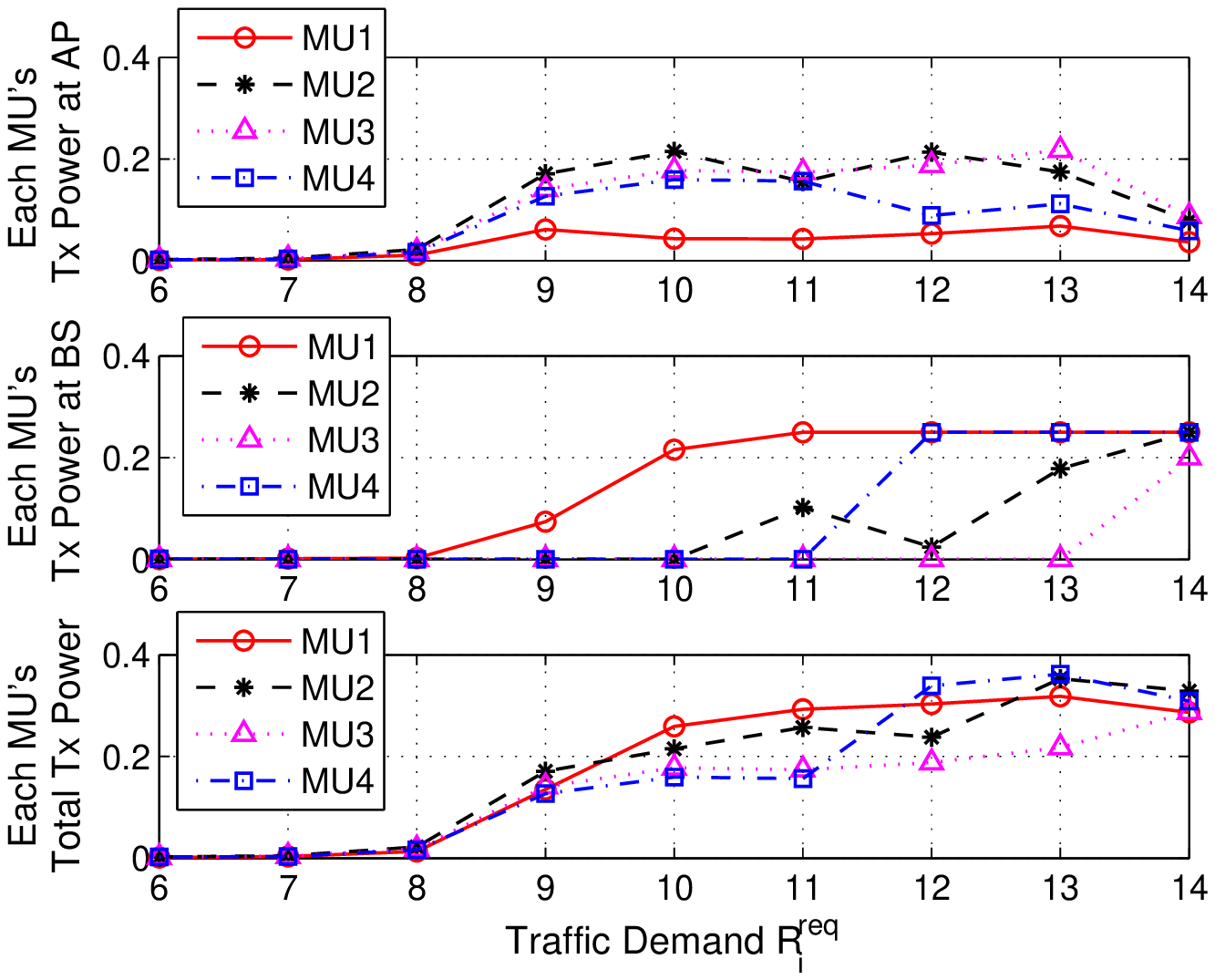}}
\vspace{-0.1in}
\caption{Detailed changes of each MU's traffic scheduling and the corresponding transmit-power allocations.}
\label{Figure_Offloading_Example}
\end{figure}

As shown in the top subplot, when the MUs' traffic demands are low (i.e., $R_{i}^{\text{req}}\leq 8\text{Mbps}$), all demands are offloaded to the AP, yielding the complete-offloading situation (i.e., Proposition \ref{proposition_special}). Accordingly, the bottom subplot shows the offloading-ratio equal to 1. In fact, we can verify that when $R_{i}^{\text{req}}\leq 8\text{Mbps}$, the conditions in Proposition \ref{proposition_special} are satisfied. In the complete-offloading situation, the MUs' transmit-powers to the AP are usually very small, which correspond to the enlarged view in the middle subplot.

However, when the MUs' traffic demands further increase (i.e., $R_{i}^{\text{req}} \geq 9\text{Mbps}$), the mutual interference among the MUs becomes significant. As a result, the MUs' transmit-power constraints can no longer afford offloading all demands to the AP. In fact, we can verify that the two conditions in Proposition \ref{proposition_special} fail to hold when $R_{i}^{\text{req}} \geq 9\text{Mbps}$. As a result, some MU starts to deliver its traffic to the BS (as shown in the top subplot), and the offloading ratio starts to decrease (as shown in the bottom subplot). Correspondingly, the MU needs to allocate their transmit-powers to the AP and BS for accommodating its traffic scheduling, and the MUs' transmit-powers to the BS increases (as shown in the middle subplot).

To gain deeper insights of the results in Fig. \ref{Figure_Offloading_Ratio}, we further plot each individual MU's traffic scheduling in Fig. \ref{Figure_Offloading_MUThroughput} and the transmit-powers in Fig. \ref{Figure_Offloading_MUPower}. The two figures show that the MUs' traffic offloading decisions and the transmit-power allocations are strongly correlated due to the mutual interference at the AP. In particular, we observe that as the MUs' traffic demands increase (i.e., $R_{i}^{\text{req}}\geq 8$Mbps), MU $1$ (whose channel gain $g_{1B}$ is the largest) first starts to re-direct its traffic to the BS, and then MU $2$ (whose $g_{2B}$ is the second largest) follows. Such a result is consistent with the intuition, because a larger channel gain to the BS requires a smaller transmit-power to achieve the same data rate to the BS. The middle and bottom subplots of Fig. \ref{Figure_Offloading_MUPower} also show that each MU's transmit-power limit to the BS (i.e., $P_{iB}^{\max}=0.25$W) and total power budget (i.e., $P_{i}^{\max}=0.25$W) eventually become tight as the MUs' demands increase. That is why Problem (CMP) is infeasible when $R_i^{\text{req}}=15$Mbps.

\subsection{Advantages of Proposed Traffic Offloading}
Figure \ref{Figure_Offloading_EconomicCost} compares the performances of the proposed traffic offloading with two other schemes, i.e., the \textit{zero-offloading scheme} and the \textit{fixed-offloading scheme}. In the zero-offloading scheme, no MU's traffic is offloaded to the AP; while in the fixed-offloading scheme, each MU offloads 50\% of its demand to the AP. We use the 4-MU and 8-MU scenarios used before to evaluate the performances of these two schemes. Figure \ref{Figure_Offloading_EconomicCost} validates that the total cost can be greatly reduced by using our proposed traffic offloading scheme\footnote{We use the result of Algorithm (Dis) to denote the proposed traffic offloading scheme in this figure, since the performances of Algorithm (Cen) and Algorithm (Dis) are very close.}. In both subplots, the proposed offloading can save more than 75\% of the total cost obtained by the zero-offloading scheme, and more than 65\% of the total cost of the fixed-offloading scheme.

\begin{figure}[tbph]
  \centering
  \includegraphics[scale=0.5]{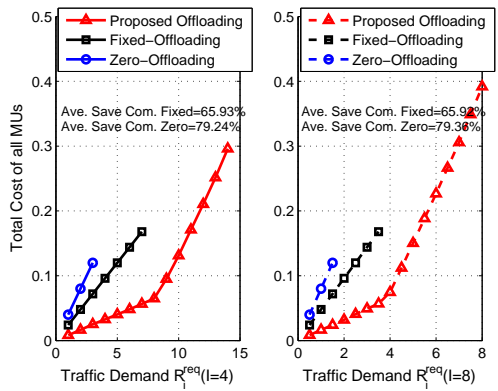}\\
  \vspace{-0.2in}
  \caption{Performance comparisons among the proposed traffic offloading (obtained by Algorithm (Dis)), the zero-offloading scheme, and the fixed-offloading scheme. Left: the 4-MU scenario. Right: the 8-MU scenario} \label{Figure_Offloading_EconomicCost}
\end{figure}

Figure \ref{Figure_Offloading_EconomicCost} also validates that the proposed offloading can increase the network capacity in terms of accommodating more MUs' traffic demands. Specifically, the left subplot shows that the zero-offloading cannot support a demand more than 3Mbps for each MU, and the fixed-offloading cannot support a demand more than 7Mbps for each MU. In comparison, the proposed traffic offloading can accommodate each MU's demand up to $R_{i}^{\text{req}}=14$Mbps. Similar results also appear in the right subplot.

\subsection{Performance under Different Locations of the MUs}
We further evaluate the performance of the proposed offloading under different locations of the MUs. To this end, we vary the center of the circle within which the MUs are randomly located (as described at the beginning of Section \ref{section_numerical}) according to $(170\text{m},0\text{m}), (220\text{m},0\text{m}),(270\text{m},0\text{m})$ and $(320\text{m},0\text{m})$. This corresponds to that the MUs are gradually moving away from the BS and closer to the AP. For each location, we independently and randomly generate 100 different sets of the MUs' locations. Figure \ref{Figure_Offloading_Locations} shows the corresponding average results, in which we plot the results for three cases, i.e., each MU's demand $R_i^{\text{req}}=4,8,$ and $12$Mbps. Subplot \ref{Fig.AP} shows the average total traffic offloaded to the AP versus different locations. It shows that when the MUs are closer to the AP, more traffic demands are offloaded to the AP, which is attributed to the better channel gains between the MUs and the AP. Subplot \ref{Fig.BS} shows that less traffic are delivered through the BS as MUs move closer to the AP. Finally, Subplot \ref{Fig.TotalCost} shows the total cost decreases when the MUs are closer to the AP due to effective offloading.

\vspace{-0.2in}
\begin{figure}[tbph]
\centering
\subfigure[Ave. Throughput at the AP]{
\label{Fig.AP}
\includegraphics[width=4.6cm,height=3.6cm]{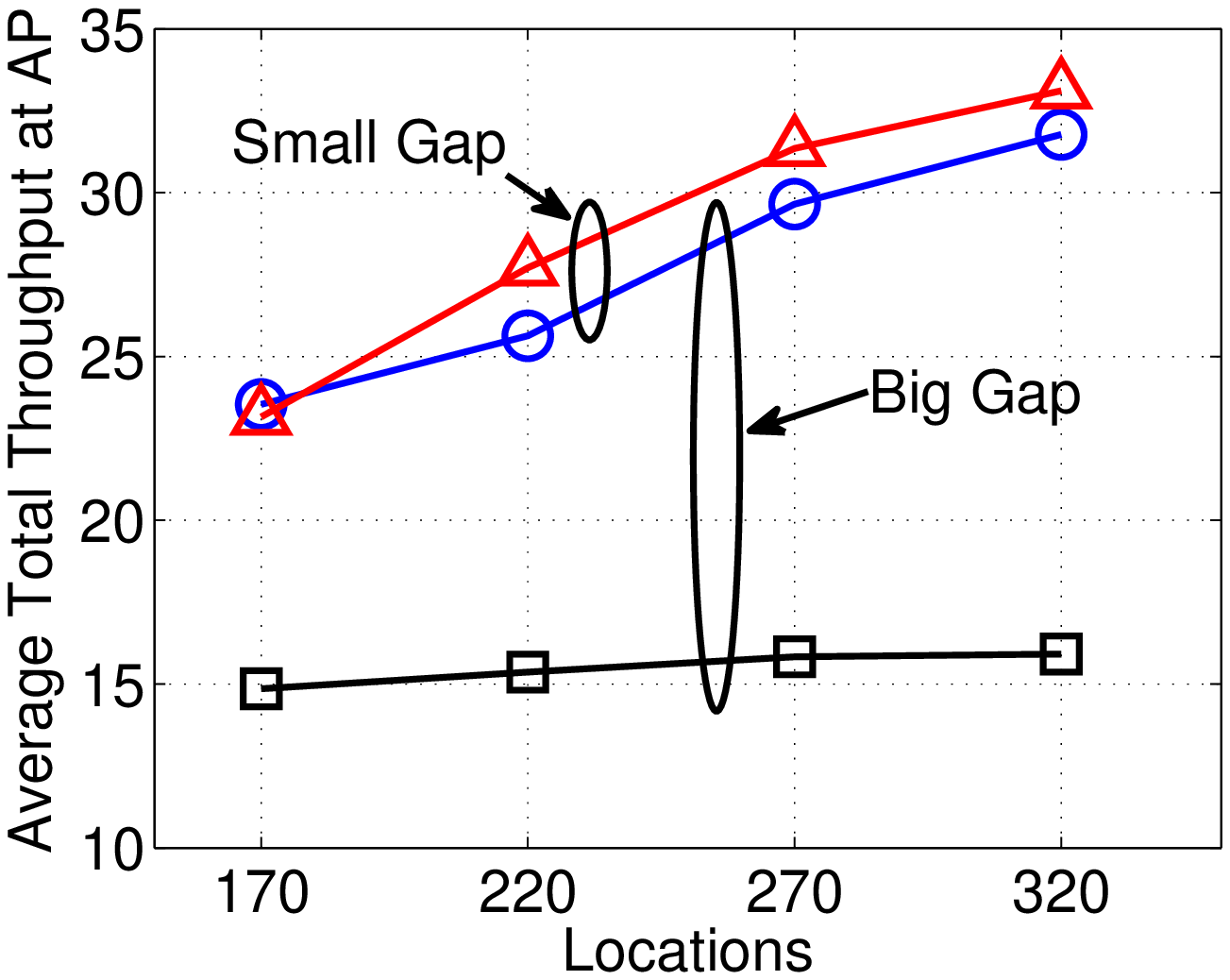}}
\subfigure[Ave. Throughput at the BS]{
\label{Fig.BS}
\includegraphics[width=4.6cm,height=3.6cm]{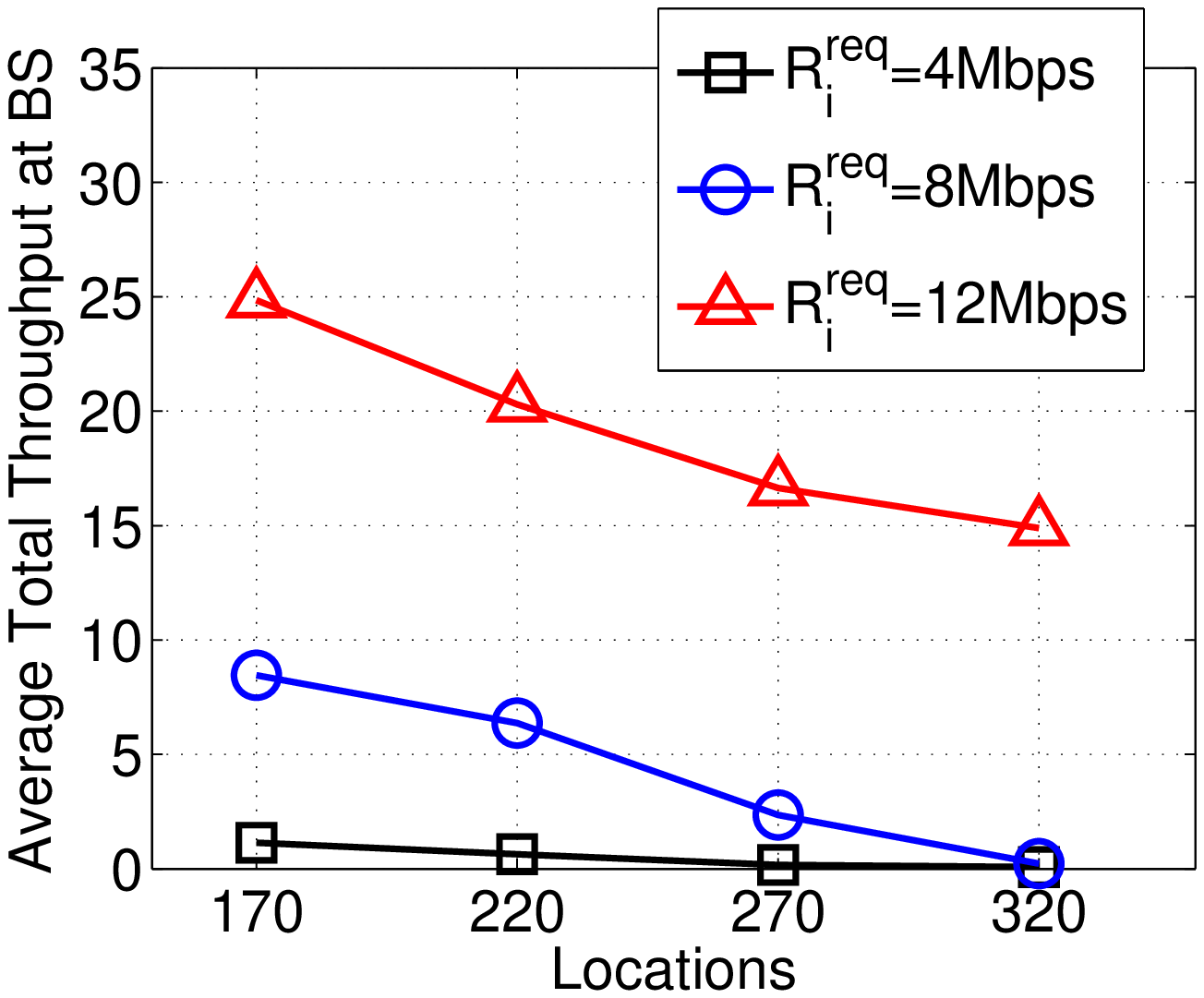}}
\subfigure[Ave. Total Cost]{
\label{Fig.TotalCost}
\includegraphics[width=4.6cm,height=3.6cm]{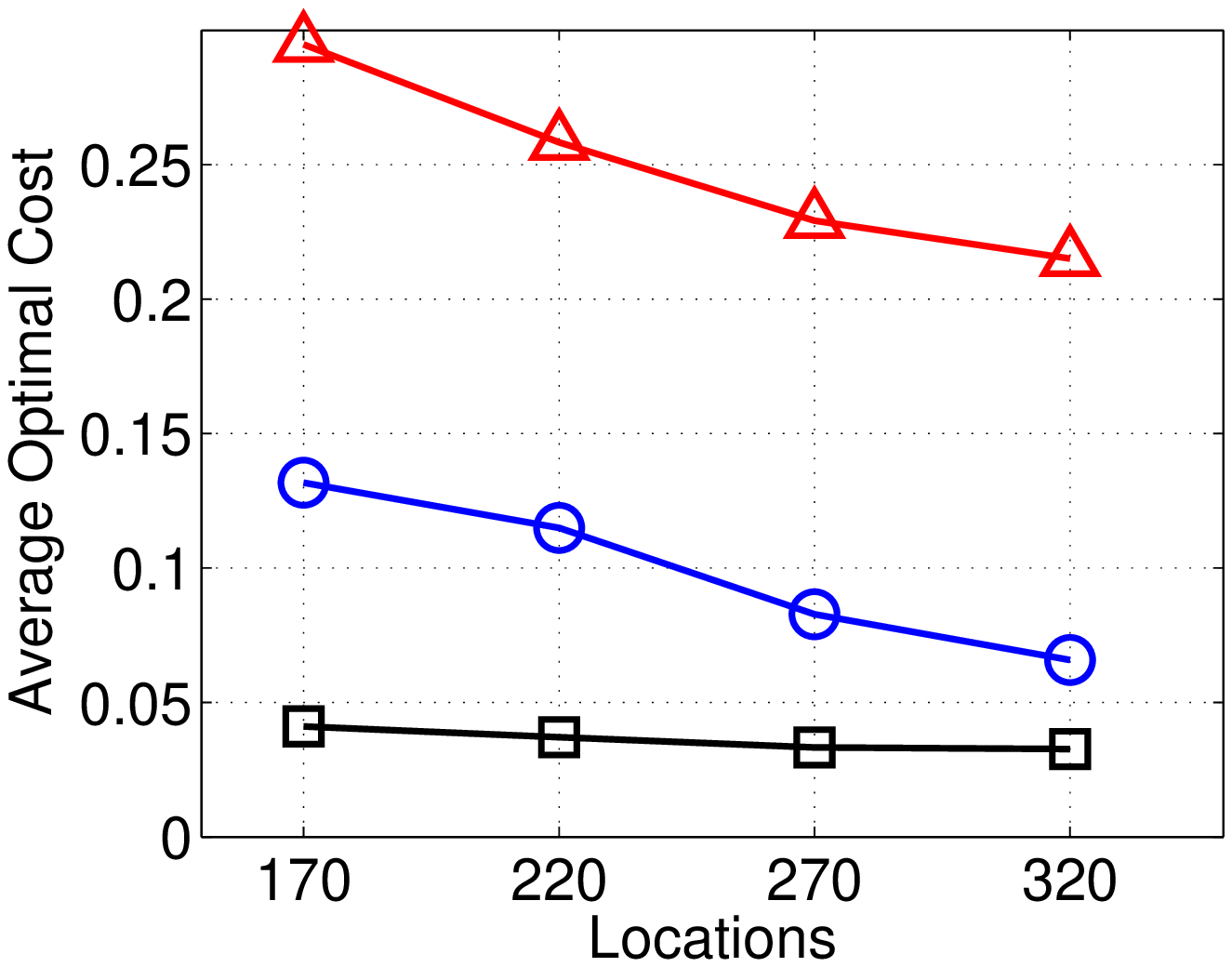}}
\vspace{-0.1in}
\caption{Performance of the traffic offloading under different locations of the MUs.}
\label{Figure_Offloading_Locations}
\end{figure}

An interesting observation in Subplot \ref{Fig.AP} is that the difference of offloaded traffic is marginal when the traffic demand increases from $R_{i}^{\text{req}}=8\text{Mbps}$ (the line marked with circles) to $R_{i}^{\text{req}}=12\text{Mbps}$ (the line marked with triangles). However, the difference is very significant when the demand increases from $R_{i}^{\text{req}}=4\text{Mbps}$ (the line marked with squares) to $R_{i}^{\text{req}}=8\text{Mbps}$. The result is consistent with the intuition that offloading traffic to AP becomes less attractive due to the heavy interference at the AP when the MUs' traffic demands are high. In fact, as shown in Subplot \ref{Fig.BS}, most traffic demands are delivered to the BS when each MU's demand increases from $8\text{Mbps}$ to $12\text{Mbps}$.

\subsection{Evaluating the Impact of $\Delta_{\text{top}}$ used in Algorithms (Cen) and (Dis)}

\vspace{-0.1in}
\begin{table}[htbp]
        \caption{Optimal Cost and Computational Time with Different $\Delta_{\text{top}}$ (Algorithm (Cen))}
        \vspace{-0.5in}
        \begin{center}
        \begin{tabular}{|c||c|c|c|c|c|c|}
  \hline
 & $R_i^{\text{req}}=3\text{Mbps}$ &$R_i^{\text{req}}=4\text{Mbps}$&$R_i^{\text{req}}=5\text{Mbps}$&$R_i^{\text{req}}=6\text{Mbps}$&$R_i^{\text{req}}=7\text{Mbps}$&$R_i^{\text{req}}=8\text{Mbps}$\\
 \hline
$\Delta_{\text{top}}=0.01$&0.0253,\text{~~}34.2s&0.0341,\text{~~}36.1s&0.0418,\text{~~}39.1s&0.0511,\text{~~}40s&0.0565,\text{~~}37.4s&0.0666,\text{~~}37.3s\\
\hline
$\Delta_{\text{top}}=0.005$&0.0253,\text{~~}64.3s&0.0327,\text{~~}69.1s&0.0404,\text{~~}74.9s&0.0497,\text{~~}78.7s&0.0565,\text{~~}73.1s&0.0651,\text{~~}72.3s\\
\hline
$\Delta_{\text{top}}=0.0025$&0.0246,\text{~~}106.1s&0.0327,\text{~~}112.1&0.0404,\text{~~}122.4s&0.0490,\text{~~}173.9s&0.0565,\text{~~}157.9s&0.0643,\text{~~}156.4s\\
\hline
$\Delta_{\text{top}}=0.001$&0.0245,\text{~~}324.9s&0.0325,\text{~~}341.1s&0.0404,\text{~~}370.5s&0.0486,\text{~~}393.5s&0.0565,\text{~~}365.4s&0.0645,\text{~~}358.9s\\
\hline
& $R_i^{\text{req}}=9\text{Mbps}$ &$R_i^{\text{req}}=10\text{Mbps}$&$R_i^{\text{req}}=11\text{Mbps}$&$R_i^{\text{req}}=12\text{Mbps}$&$R_i^{\text{req}}=13\text{Mbps}$&$R_i^{\text{req}}=14\text{Mbps}$\\
\hline
$\Delta_{\text{top}}=0.01$&0.0965,\text{~~}36.6s&0.1328,\text{~~}32.7s&0.1744,\text{~~}31.4s&0.2158,\text{~~}29.9s&0.2565,\text{~~}30.1s&0.2967,\text{~~}30.1s\\
\hline
$\Delta_{\text{top}}=0.005$&0.0950,\text{~~}69.9s&0.1315,\text{~~}62.9s&0.1728,\text{~~}59.8s&0.2142,\text{~~}57.6s&0.2547,\text{~~}57.4s&0.2967,\text{~~}56.7s\\
\hline
$\Delta_{\text{top}}=0.0025$&0.0950,\text{~~}148.9s&0.1315,\text{~~}135.5s&0.1728,\text{~~}125.5s&0.2142,\text{~~}121.9s&0.2547,\text{~~}123.9s&0.2959,\text{~~}122.7s\\
\hline
$\Delta_{\text{top}}=0.001$&0.0950,\text{~~}348.1s&0.1315,\text{~~}317.2s&0.1722,\text{~~}300.3s&0.2133,\text{~~}292.1s&0.2547,\text{~~}288.6s&0.2957,\text{~~}286.9s\\
\hline
\end{tabular}
\end{center}
\label{table_monotonic}
\end{table}
Tables \ref{table_monotonic} and \ref{table_distributed} validate the use of $\Delta_{\text{top}}=0.005$ in Algorithms (Cen) and (Dis) for the linear search. We show that using an $\Delta_{\text{top}}$ smaller than 0.005 yields a very limited improvement on the performance but incurs a significant increase in the computational time. To this end, we use the 4-MU scenario, and try $\Delta_{\text{top}}=0.001,0.0025,0.005$ and $0.01$ in Algorithms (Cen) and (Dis). We show the obtained total cost (the first number in each cell) and the computational time (the second number in each cell) versus different traffic demands. As shown in Table \ref{table_monotonic}, using $\Delta_{\text{top}}=0.005$ in Algorithm (Cen) yields a small average loss of 0.385\% compared to using $\Delta_{\text{top}}=0.001$, while it can reduce $78.76\%$ of the computational time. Similarly, compared to $\Delta_{\text{top}}=0.0025$, using $\Delta_{\text{top}}=0.005$ yields an average loss equal to $0.288\%$, while reducing $46.26\%$ of the computational time. Table \ref{table_distributed} shows similar conclusions for Algorithm (Dis).

\begin{table}[htbp]
        \caption{Optimal Cost and Computational Time with Different $\Delta_{\text{top}}$ (Algorithm (Dis))}
        \vspace{-0.5in}
        \begin{center}
        \begin{tabular}{|c||c|c|c|c|c|c|}
  \hline
 & $R_i^{\text{req}}=3\text{Mbps}$ &$R_i^{\text{req}}=4\text{Mbps}$&$R_i^{\text{req}}=5\text{Mbps}$&$R_i^{\text{req}}=6\text{Mbps}$&$R_i^{\text{req}}=7\text{Mbps}$&$R_i^{\text{req}}=8\text{Mbps}$\\
\hline
$\Delta_{\text{top}}=0.01$&0.0252,\text{~~}0.68s&0.0340,\text{~~}0.71s&0.0417,\text{~~}0.84s&0.0482,\text{~~}1.01s&0.0564,\text{~~}1.01s&0.0664,\text{~~}1.04s\\
\hline
$\Delta_{\text{top}}=0.005$&0.0252,\text{~~}1.45s&0.0327,\text{~~}1.53s&0.0403,\text{~~}1.72s&0.0482,\text{~~}2.04s&0.0564,\text{~~}1.89s&0.0649,\text{~~}1.87s\\
\hline
$\Delta_{\text{top}}=0.0025$&0.0245,\text{~~}2.93s&0.0320,\text{~~}3.29s&0.0403,\text{~~}3.81s&0.0482,\text{~~}3.93s&0.0564,\text{~~}4.17s&0.0643,\text{~~}4.23s\\
\hline
$\Delta_{\text{top}}=0.001$&0.0242,\text{~~}7.22s&0.0321,\text{~~}7.81s&0.0400,\text{~~}8.52s&0.0482,\text{~~}9.21s&0.0561,\text{~~}9.87s&0.0640,\text{~~}9.36s\\
\hline
& $R_i^{\text{req}}=9\text{Mbps}$ &$R_i^{\text{req}}=10\text{Mbps}$&$R_i^{\text{req}}=11\text{Mbps}$&$R_i^{\text{req}}=12\text{Mbps}$&$R_i^{\text{req}}=13\text{Mbps}$&$R_i^{\text{req}}=14\text{Mbps}$\\
\hline
$\Delta_{\text{top}}=0.01$&0.0964,\text{~~}0.91s&0.1326,\text{~~}0.87s&0.1725,\text{~~}0.73s&0.2118,\text{~~}0.64s&0.2531,\text{~~}0.62s&0.2960,\text{~~}0.65s\\
\hline
$\Delta_{\text{top}}=0.005$&0.0950,\text{~~}1.87s&0.1313,\text{~~}1.54s&0.1710,\text{~~}1.41s&0.2104,\text{~~}1.39s&0.2516,\text{~~}1.28s&0.2960,\text{~~}1.29s\\
\hline
$\Delta_{\text{top}}=0.0025$&0.0950,\text{~~}3.73s&0.1313,\text{~~}3.24s&0.1710,\text{~~}2.98s&0.2104,\text{~~}2.91s&0.2516,\text{~~}2.54s&0.2951,\text{~~}2.51s\\
\hline
$\Delta_{\text{top}}=0.001$&0.0949,\text{~~}8.52s&0.1313,\text{~~}7.94s&0.1701,\text{~~}7.17s&0.2104,\text{~~}6.71s&0.2516,\text{~~}6.18s&0.2951,\text{~~}6.35s\\
\hline
\end{tabular}
\end{center}
\label{table_distributed}
\end{table}

\section{Conclusion}
\label{section_conclusion}
In this paper, we investigated the joint optimization of traffic scheduling and transmit-power allocations for MUs' traffic offloading with dual-connectivity. In spite of the non-convexity of joint optimization problem, we proposed two efficient algorithms, namely, the centralized algorithm and the distributed algorithm, to minimize the MUs' total cost. The numerical results validated the effectiveness and the computational efficiency of the proposed algorithms, and showed that the proposed traffic offloading can significantly reduce the MUs' total cost.

Regarding the future work, we will consider the more general case that the MUs' traffic demands requirements might not be feasible, and propose an admission control scheme to ensure the feasibility of Problem (CMP). Moreover, we will consider multiple coexisting APs and explore the possibility of allowing different MUs to offload traffic through different APs.


\section*{Appendix I: Proof of Proposition \ref{proposition_power_theta}}
\label{proof_propos1}
Based on (\ref{eq_sinr}), we first show that given $\{\theta_{i}\}_{i\in\mathcal{I}}$, $\{p_{iA}\}_{i\in\mathcal{I}}$ can be uniquely given by (\ref{eq_piA_thetaiA}). The details are as follows. We introduce a variable $T=\sum_{i\in\mathcal{I}} {p_{iA}g_{iA}}+n_A$.
Using $T$ and (\ref{eq_sinr}), we can derive $\theta_i=\frac{p_{iA}g_{iA}}{T-p_{iA}g_{iA}}$
which can be translated into $p_{iA}g_{iA}=T\frac{\theta_i}{1+\theta_i},\forall i\in\mathcal{I}$. By summarizing this equation on two sides over all MUs and performing some manipulations, we obtain $T=\frac{n_A}{1-\sum_{i\in\mathcal{I}} {\frac{\theta_i}{1+\theta_i}}}$, which leads to (\ref{eq_piA_thetaiA}). Notice that the condition $\sum_{i\in\mathcal{I}} \frac{\theta_i}{1+\theta_i}<1$ is required, such that the obtained $T$ is consistent with its physical meaning.

Then, we continue to show that, given $\{p_{iA}\}_{i\in\mathcal{I}}$ feasible for Problem (TPA-P), $\{\theta_{iA}\}_{i\in\mathcal{I}}$ given by (\ref{eq_sinr}) always yields $\sum_{i\in\mathcal{I}} \frac{\theta_i}{1+\theta_i}<1$. The details are as follows.
\underline{\textit{First}}, let us consider a profile $\{p_{iA}\}_{i\in\mathcal{I}}$ feasible for Problem (TPA-P). In particular, in this $\{p_{iA}\}_{i\in\mathcal{I}}$, there exists at least two different MUs allocating positive transmit-powers to the AP. We thus can obtain $\{\theta_{i}\}_{i\in\mathcal{I}}$ based on (\ref{eq_sinr}). Such $\{\theta_{i}\}_{i\in\mathcal{I}}$ always yields $\sum_{i\in\mathcal{I}} \frac{\theta_i}{1+\theta_i}<1$. Otherwise, some MU's transmit-power must be negative according to the reversed mapping (\ref{eq_piA_thetaiA}). \underline{\textit{Second}}, let us consider another $\{p_{iA}\}_{i\in\mathcal{I}}$ feasible for Problem (TPA-P), in which only one MU (e.g., MU $i$) allocating positive transmit-power to the AP. Then, the resulting $\theta_i$ (based on (\ref{eq_sinr})) can always ensure $\sum_{i\in\mathcal{I}} \frac{\theta_i}{1+\theta_i}<1$. In summary, any $\{p_{iA}\}_{i\in\mathcal{I}}$ feasible for Problem (TPA-P) will yield $\{\theta_{i}\}_{i\in\mathcal{I}}$ (based on (\ref{eq_sinr})) which can ensure
$\sum_{i\in\mathcal{I}} \frac{\theta_i}{1+\theta_i}<1$.

\section*{Appendix II: Proof of Proposition \ref{proposition_equivalent_M0}}
\label{proof_propos3}
By using $\rho_i$ to replace $\theta_i$, Problem (SINR-P) can be transformed into the following equivalent form
\begin{eqnarray}
\text{(P1): }  && \text{Maximize} \sum_{i\in\mathcal{I}}(\pi_B-\pi_A)W \log_2\left(\frac{1}{1-\rho_i}\right) \nonumber\\
&& \text{Subject to: } 0\leq \frac{\rho_i}{1-\rho_i} \leq 2^{\frac{R_i^{\text{req}}}{W}}-1,\forall i\in\mathcal{I},\label{P4_Con_throughput_Mform}\\
&& \text{~~~~~~~~~~~~}  \frac{n_A}{g_{iA}}\rho_i\frac{1}{1-\sum_{i\in\mathcal{I}}\rho_i}\leq P_{iA}^{\max},\forall i\in\mathcal{I}, \label{P4_Con_PiAmax_Mform}\\
&& \text{~~~~~~~~~~~~}  \frac{n_B}{g_{iB}}2^{\frac{R_i^{\text{req}}}{B}} (1-\rho_i)^{\frac{W}{B}} \leq P_{iB}^{\max} + \frac{n_B}{g_{iB}},\forall i\in\mathcal{I}, \label{P4_Con_PiBmax_Mform} \\
&& \text{~~~~~~~~~~~~}  \frac{n_A}{g_{iA}}\rho_i\frac{1}{1-\sum_{i\in\mathcal{I}}\rho_i}
+\frac{n_B}{g_{iB}}2^{\frac{R_i^{\text{req}}}{B}} (1-\rho_i)^{\frac{W}{B}} \leq P_i^{\max} +\frac{n_B}{g_{iB}},\forall i\in\mathcal{I}, \label{P4_Con_Pimax_Mform} \\
&& \text{~~~~~~~~~~~~} \sum_{i\in\mathcal{I}}\rho_i < 1, \label{P4_Con_SumM_Mform}  \\
&& \text{Variables: } \rho_i, \forall i\in\mathcal{I}.  \nonumber
\end{eqnarray}
By using $\rho_0$ that ensures $\rho_0+\sum_{i\in\mathcal{I}}\rho_i=1$, we can obtain the form of Problem (SINR-M-P) exactly.

We next prove Proposition \ref{proposition_equivalent_M0}. In particular, let $\rho_0^\ast$ and $\{\rho_i^\ast\}_{i\in\mathcal{I}}$ together denote an optimal solution of Problem (SINR-M-P). To prove Proposition \ref{proposition_equivalent_M0}, it suffices to show that $\{\rho_i^\ast\}_{i\in\mathcal{I}}$ corresponds to an optimal solution of Problem (P1) (notice that since $\rho_i$ and $\theta_i$ form a one-to-one mapping given in (\ref{eq_M_theta}), the equivalence between Problem (SINR-P) and Problem (P1) directly follows). We thus prove that $\{\rho_i^\ast\}_{i\in\mathcal{I}}$ corresponds to an optimal solution of Problem (P1). Our key idea is to illustrate that the feasible region of Problem (P1) and that of Problem (SINR-M-P) form a one-to-one mapping exactly, if Problem (SINR-P) (i.e., Problem (P1)) is feasible. The details are illustrated by the following three points.
\begin{itemize}
\item \textit{First}, Problem (SINR-P) and Problem (P1) are feasible, since they are based on the series of equivalent transformations from Problem (CMP). The feasibility of Problem (P1) enables us to introduce a $\rho_0>0$ such that $\rho_0+\sum_{i\in\mathcal{I}} {\rho_i}= 1$.
\item \textit{Second}, let us denote the feasible region of Problem (P1) by $\mathcal{R}$, i.e., the intersection of (\ref{P4_Con_throughput_Mform}), (\ref{P4_Con_PiAmax_Mform}), (\ref{P4_Con_PiBmax_Mform}), (\ref{P4_Con_Pimax_Mform}), and (\ref{P4_Con_SumM_Mform}). $\mathcal{R}$ is nonempty. For each point $\{\rho_i\}_{i\in\mathcal{I}}\in\mathcal{R}$, let us set $\rho_0=1-\sum_{i\in\mathcal{I}}\rho_i$. Then, point $(\rho_0,\{\rho_i\}_{i\in\mathcal{I}})$ is compatible with (\ref{P5_Con_throughput}), (\ref{P5_Con_PiAmax}), (\ref{P5_Con_PiBmax}), (\ref{P5_Con_Pimax}), and (\ref{P5_Con_M}). In other words, Problem (SINR-M-P) is feasible under such a given point $(\rho_0,\{\rho_i\}_{i\in\mathcal{I}})$ with $\{\rho_i\}_{i\in\mathcal{I}}\in\mathcal{R}$. Thus, Problem (SINR-M-P) is feasible.
\item \textit{Third}, let us define $\mathcal{R}'$ as the feasible set of Problem (SINR-M-P) and consider any point $(\rho_0,\{\rho_i\}_{i\in\mathcal{I}})\in\mathcal{R}'$. Then, based on (\ref{P5_Con_M}), such a $\{\rho_i\}_{i\in\mathcal{I}}$ is compatible with (\ref{P4_Con_throughput_Mform}), (\ref{P4_Con_PiAmax_Mform}), (\ref{P4_Con_PiBmax_Mform}), (\ref{P4_Con_Pimax_Mform}), and (\ref{P4_Con_SumM_Mform}), i.e., $\{\rho_i\}_{i\in\mathcal{I}}$ belongs to $\mathcal{R}$.
\end{itemize}

Summarizing the above three points, there exists a one-to-one mapping between each point in the feasible region $\mathcal{R}'$ of Problem (SINR-M-P) and each point in the feasible region $\mathcal{R}$ of Problem (P1). Moreover, the objective function of Problem (SINR-M-P) is equal to that of Problem (P1) under (\ref{P5_Con_M}). Hence, if $(\rho_0^\ast,\{\rho_i^\ast\}_{i\in\mathcal{I}})$ is an optimal solution of Problem (SINR-M-P), $\{\rho_i^\ast\}_{i\in\mathcal{I}}$ ensures to maximize the objective function of Problem (P1) while being feasible.

\section*{Appendix III: Proof of Proposition \ref{proposition_P5R}}
\label{proof_propos4}
The proof is based on showing contradiction. Suppose that $\sum_{i\in\mathcal{I}} \rho_{i,(\rho_0)}^{\ast,\text{sub}}> 1-\rho_0$. Then, we can select an MU $i$ and slightly reduce $\rho_{i,(\rho_0)}^{\ast,\text{sub}}$. Consequently, the objective function can be reduced without violating any constraint (since (\ref{P5_Con_throughput}), (\ref{P5_Con_PiAmax}), (\ref{P5_Con_PiBmax}), and (\ref{P5_Con_Pimax}) in $\mathcal{G}_{(\rho_0)}$ are all increasing in $\rho_i$). As a result, the constraint in $\mathcal{H}_{(\rho_0)}$ should be always binding when achieving the optimum, i.e., being consistent with (\ref{P5_Con_M}).

\section*{Appendix IV: Feasibility Check of Problem (SINR-M-SubP)}
\label{feasibility_check}
We present a scheme for checking the feasibility of Problem (SINR-M-SubP). Specifically, under a given $\rho_0$, we first consider the following optimization problem.
\begin{eqnarray}
\text{(P2): } L_{(\rho_0)}=\text{Maximize } \rho_0+\sum_{i\in\mathcal{I}}\rho_i, \text{ Subject to: constraints } (\ref{P5_Con_throughput})-(\ref{P5_Con_Pimax}), \text{ Variables: } \rho_i,\forall i\in\mathcal{I}. \nonumber
\end{eqnarray}
In Problem (P2), $\rho_0$ is given in (\ref{P5_Con_throughput}), (\ref{P5_Con_PiAmax}), and (\ref{P5_Con_PiBmax}). If $L_{(\rho_0)}\geq 1$, then the intersection of (\ref{P5_Con_throughput}), (\ref{P5_Con_PiAmax}), (\ref{P5_Con_PiBmax}), (\ref{P5_Con_Pimax}), and (\ref{P5_Con_M}) will be non-empty, meaning that Problem (SINR-M-SubP) is feasible. Otherwise (i.e., $L_{(\rho_0)}<1$), Problem (SINR-M-SubP) is infeasible. To obtain $L_{(\rho_0)}$, we need to solve Problem (P2), which corresponds to a monotonic optimization problem as follows:
\begin{eqnarray}
\text{(P2-E): }  L_{(\rho_0)}=\text{Maximize } \rho_0+\sum_{i\in\mathcal{I}}\rho_i,
\text{ Subject to: } \{\rho_i\}_{i\in\mathcal{I}}\in\mathcal{G}_{(\rho_0)} \cap \mathcal{H}'_{(\rho_0)}, \text{ Variables: } \rho_i,\forall i\in\mathcal{I}. \nonumber
\end{eqnarray}
The normal set $\mathcal{G}_{(\rho_0)}$ is given in (\ref{eq_normalset}), and the reversed normal set $\mathcal{H}'_{(\rho_0)}$ is given by $\mathcal{H}'_{(\rho_0)}=\big\{\{\rho_i\}_{i\in\mathcal{I}}| 0\leq  \rho_i\leq  1 - \rho_0,\forall i\in\mathcal{I}\big\}$. Based on the monotonicity of Problem (P2-E), we can also design a poly-block approximation based algorithm (which is similar to Algorithm (Cen-Sub) in Sec. \ref{subsection_subproblem}) to solve Problem (P2-E) and obtain $L_{(\rho_0)}$. Due to space limitation, we skip the details of this algorithm. In summary, a sketch (i.e., Algorithm (Cen-Sub-FC)) to check the feasibility of Problem (SINR-M-SubP) is provided.
\begin{algorithm}
\textbf{Algorithm (Cen-Sub-FC): A Sketch to Check the Feasibility of Problem (SINR-M-SubP)}

\vspace{-0.1in}

\hrulefill

\begin{algorithmic}[1]
\STATE \textbf{Input:} the currently given $\rho_0$.
\STATE Solve Problem (P2-E) with the polyblock-approximation algorithm and obtain the value of $L_{(\rho_0)}$.
\IF{$L_{(\rho_0)}\geq 1$}
\STATE Output: Problem (SINR-M-SubP) is feasible.
\ELSE
\STATE Output: Problem (SINR-M-SubP) is infeasible.
\ENDIF
\end{algorithmic}
\end{algorithm}

\section*{Appendix V: Proof of Lemma \ref{lemma_diffcases}}
\label{proof_lemma2}
We use MU $i$ as an example in this prove. $J_i(\rho_i)$ is convex under Assumption \ref{assumption_bandwidth}, i.e.,
$J_i'(\rho_i)=\frac{n_A}{g_{iA}}\frac{1}{\rho_0}-\frac{n_B}{g_{iB}}2^{\frac{R_i^{\text{req}}}{B}}\frac{W}{B}(1-\rho_i)^{\frac{W}{B}-1}$
is increasing. Thus, if $J_i'(\rho_i)_{|\rho_i=0}> 0$ (which yields Condition (C3)), $J_i(\rho_i)$ is increasing. Otherwise (i.e., Condition (C3) does not hold), there exists a unique $\chi_{i,(\rho_0)}\in(0,1]$ that yields $J_i'(\rho_i)_{|\rho_i=\chi_{i,(\rho_0)}}=0$. By solving $J_i'(\rho_i)=0$, we obtain $\chi_{i,(\rho_0)}$ given in (\ref{value_Z}). Thus, $J_i'(\rho_i)<0$ (i.e., $J_i(\rho_i)$ is decreasing) for $0<\rho_i<\chi_{i,(\rho_0)}$ and $J_i'(\rho_i)>0$ (i.e., $J_i(\rho_i)$ is increasing) for $\chi_{i,(\rho_0)}<\rho_i\leq 1$.

\section*{Appendix VI: Proof of Proposition \ref{proposition_diffcases}}
\label{proof_lemma3}
We first prove the results \underline{when Condition (C3) holds}. Since $J_i(\rho_i)$ is increasing when Condition (C3) holds (Lemma \ref{lemma_diffcases}), there exits \textit{three different subcases} as follows:

\noindent \underline{Subcase i) (as shown in Fig. \ref{Fig.sub1.1})} If $J_i(0)>P_i^{\max}+\frac{n_B}{g_{iB}}$, there exists no feasible $\rho_i\in(0,1]$ such that (\ref{P5_E_Con_Pimax}) can hold. As a result, we set $\left(\underline{\mu}_{i,(\rho_0)},\overline{\mu}_{i,(\rho_0)}\right)=(1,0)$ (i.e., the first case) such that the feasible interval of $\rho_i$ is empty.

\noindent \underline{Subcase ii) (as shown in Fig. \ref{Fig.sub1.2})} If $J_i(0)\leq P_i^{\max}+\frac{n_B}{g_{iB}}\leq J_i(1)$, then there exists a unique point $v\in(0,1]$ such that $J_i(v)=P_i^{\max}+\frac{n_B}{g_{iB}}$, and $J_i(\rho_i)\leq P_i^{\max}+\frac{n_B}{g_{iB}}$ when $\rho_i\in(0,v]$, which yields the second result.

\noindent \underline{Subcase iii) (as shown in Fig. \ref{Fig.sub1.3})} If $J_i(1)< P_i^{\max}+\frac{n_B}{g_{iB}}$, then (\ref{P5_E_Con_Pimax}) holds for $\rho_i\in(0,1]$, which yields the third case.

\begin{figure}[tbph]
\centering
\subfigure[Subcase i]{
\label{Fig.sub1.1}
\includegraphics[width=4.5cm,height=4.5cm]{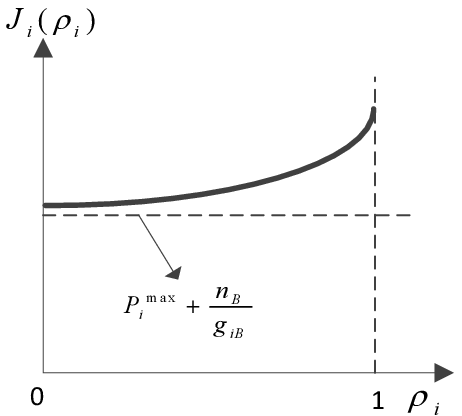}}
\subfigure[Subcase ii]{
\label{Fig.sub1.2}
\includegraphics[width=4.5cm,height=4.5cm]{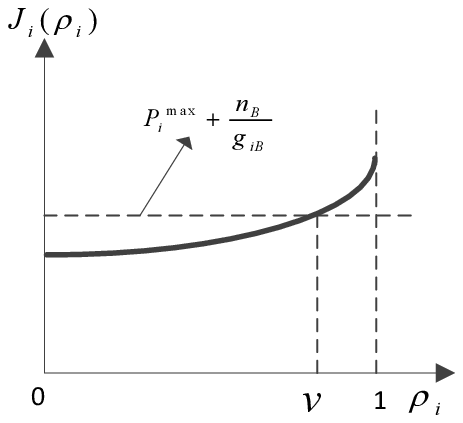}}
\subfigure[Subcase iii]{
\label{Fig.sub1.3}
\includegraphics[width=4.5cm,height=4.5cm]{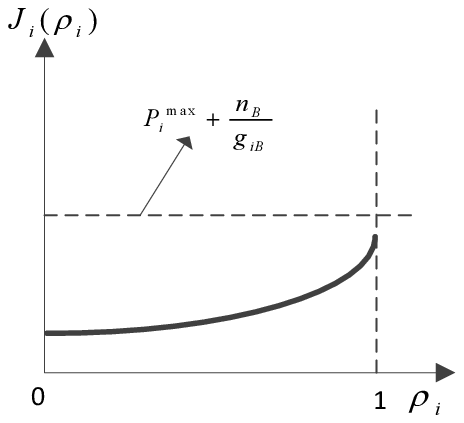}}
\caption{Illustration of three subcases when Condition (C3) holds.}
\label{Figure_Subcases1}
\end{figure}

We next prove the results \underline{when Condition (C3) does not hold}. According to Lemma \ref{lemma_diffcases}, $J_i(\rho_i)$ decreases when $\rho_i\in(0,\chi_{i,(\rho_0)}]$ and increases when $\rho_i\in[\chi_{i,(\rho_0)},1]$, with the minimum value given by $J_i(\chi_{i,(\rho_0)})$. Based on this property, there are \textit{five different subcases} as follows:

\noindent \underline{Subcase i) (as shown in Fig. \ref{Fig.sub2.1})} If $J_i(\chi_{i,(\rho_0)})>P_i^{\max}+\frac{n_B}{g_{iB}}$, there exists no feasible $\rho_i\in(0,1]$ such that (\ref{P5_E_Con_Pimax}) can hold. As a result, we set $\left(\underline{\mu}_{i,(\rho_0)},\overline{\mu}_{i,(\rho_0)}\right)=(1,0)$ (i.e., the first case) such that the feasible interval of $\rho_i$ is empty.

\noindent \underline{Subcase ii) (as shown in Fig. \ref{Fig.sub2.2})} If $J_i(\chi_{i,(\rho_0)})\leq P_i^{\max}+\frac{n_B}{g_{iB}}$, $J_i(0)\leq P_i^{\max}+\frac{n_B}{g_{iB}}$, and $J_i(1)\leq P_i^{\max}+\frac{n_B}{g_{iB}}$, then constraint (\ref{P5_E_Con_Pimax}) holds for $\rho_i\in(0,1]$, which consequently yields the second case.

\noindent \underline{Subcase iii) (as shown in Fig. \ref{Fig.sub2.3})} If $J_i(\chi_{i,(\rho_0)})\leq P_i^{\max}+\frac{n_B}{g_{iB}}$, $J_i(0)\leq P_i^{\max}+\frac{n_B}{g_{iB}}$, and $J_i(1)>P_i^{\max}+\frac{n_B}{g_{iB}}$, then there exists a special point $v\in[\chi_{i,(\rho_0)},1]$ such that $J_i(v)=P_i^{\max}+\frac{n_B}{g_{iB}}$, and (\ref{P5_E_Con_Pimax}) holds for $\rho_i\in(0,v]$. As a result, we obtain the third case.

\noindent \underline{Subcase iv) (as shown in Fig. \ref{Fig.sub2.4})} If $J_i(\chi_{i,(\rho_0)})\leq P_i^{\max}+\frac{n_B}{g_{iB}}$, $J_i(0)>P_i^{\max}+\frac{n_B}{g_{iB}}$, and $J_i(1)\leq P_i^{\max}+\frac{n_B}{g_{iB}}$, then there exists a special point $v\in(0,\chi_{i,(\rho_0)}]$ such that $J_i(v)=P_i^{\max}+\frac{n_B}{g_{iB}}$ and constraint (\ref{P5_E_Con_Pimax}) holds for $\rho_i\in[v,1]$. As a result, we obtain the fourth case.

\noindent \underline{Subcase v) (as shown in Fig. \ref{Fig.sub2.5})} If $J_i(\chi_{i,(\rho_0)})\leq P_i^{\max}+\frac{n_B}{g_{iB}}$, $J_i(0)>P_i^{\max}+\frac{n_B}{g_{iB}}$, and $J_i(1)>P_i^{\max}+\frac{n_B}{g_{iB}}$, then there exists two special points: $v_1\in(0,\chi_{i,(\rho_0)}]$ such that $J_i(v_1)=P_i^{\max}+\frac{n_B}{g_{iB}}$ and $v_2\in(\chi_{i,(\rho_0)},1]$ such that $J_i(v_2)=P_i^{\max}+\frac{n_B}{g_{iB}}$. Consequently, constraint (\ref{P5_E_Con_Pimax}) holds for $\rho_i\in[v_1,v_2]$, which yields the last case.

\begin{figure}[tbph]
\centering
\subfigure[Subcase i]{
\label{Fig.sub2.1}
\includegraphics[width=4.5cm,height=4.5cm]{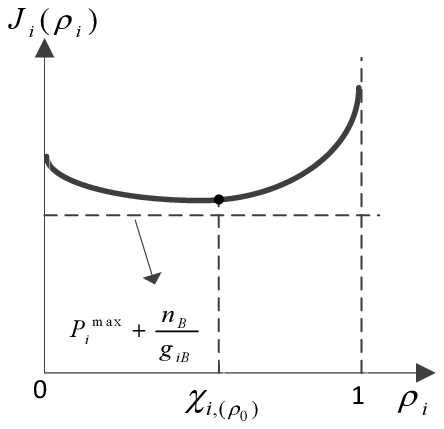}}
\subfigure[Subcase ii]{
\label{Fig.sub2.2}
\includegraphics[width=4.5cm,height=4.5cm]{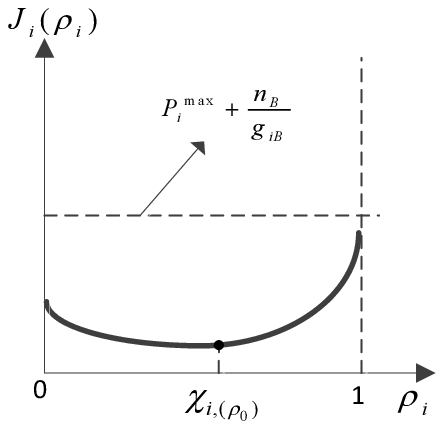}}
\subfigure[Subcase iii]{
\label{Fig.sub2.3}
\includegraphics[width=4.5cm,height=4.5cm]{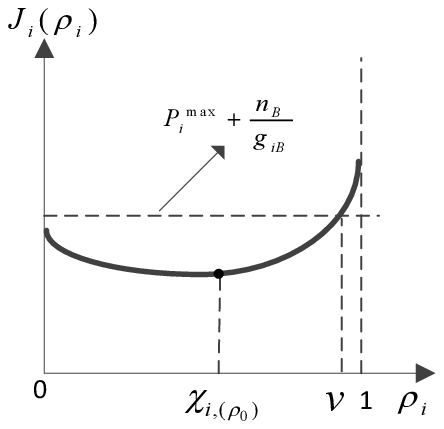}}
\subfigure[Subcase iv]{
\label{Fig.sub2.4}
\includegraphics[width=4.5cm,height=4.5cm]{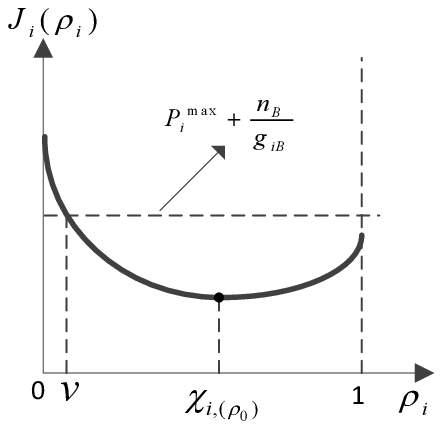}}
\subfigure[Subcase v]{
\label{Fig.sub2.5}
\includegraphics[width=4.5cm,height=4.5cm]{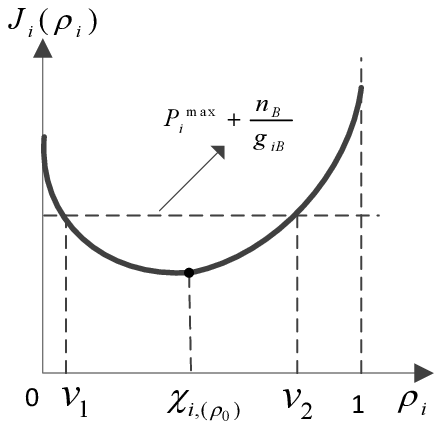}}
\caption{Illustration of five subcases when Condition (C3) does not hold.}
\label{Figure_Subcases2}
\end{figure}


\begin{thebibliography}{99}

\bibitem{Paper:LTEA} N.A. Ali, A.M. Taha, and H.S. Hassanein, ``Quality of Service in 3GPP R12 LTE-Advanced," \emph{IEEE Communications Magazine}, vol. 51, no. 8, pp. 103-109, 2013.

\bibitem{WhitePaper:Cicso} Cisco, ``Cisco Visual Networking Index: Global Mobile Data Traffic Forecast Update 2014-2019," White Paper, Feb. 2015.




\bibitem{Paper:SurveyOffloading} A. Aijaz, H. Aghvami, and M. Amani, ``A Survey on Mobile Data Offloading: Technical and Business Perspectives," \emph{IEEE Wireless Communications}, vol. 20, no. 2, pp. 104-112, April 2013.

\bibitem{Paper:SurveyOffloading2} R. Maallaw, N. Agoulmine, B. Radier, T.B. Meriem, ``A Comprehensive Survey on Offload Techniques and Management in Wireless Access and Core Networks," to appear in \emph{IEEE Communications Surveys $\&$ Tutorials}.

\bibitem{Papre:Lee} K. Lee, I. Rhee, J. Lee, S. Chong, Y. Yi, ``Mobile Data Offloading: How Much Can WiFi Deilver," in \emph{ACM CoNEXT'2010}.

\bibitem{Paper:DTN_arch} S. Dimatteo, P. Hui, B. Han, V.O.K. Li, ``Cellular Traffic Offloading through WiFi Networks," in \emph{Proc. of IEEE MASS'2011}.

\bibitem{Paper:Hossain} H. Elsawy, E. Hossain, S. Camorlinga, ``Traffic Offloading Techniques in Two-Tier Femtocell Networks," in \emph{IEEE ICC'2013}.


\bibitem{WhitePaper:Qualcomm} Qualcomm, ``LET-Advanced - Evolving and Expanding into New Frontiers," White Paper, Aug. 2014.

\bibitem{WhitePaper:NSN} Nokia Solutions and Networks, ``LTE Release 12 and Beyond," White Paper, Feb. 2014.


\bibitem{Paper:Standard} C. Sankaran, ``Data Offloading Techniques in 3GPP Rel-10 Networks: A Tutorial," \emph{IEEE Communication Magazine}, vol. 50, no. 6, pp. 46-53, June 2012.


\bibitem{Paper:Kang} X. Kang, Y.K. Chia, S. Sun, and H.F. Chong, ``Mobile Data Offloading Through A Thrid-Party WiFi Access Point: An Operator's Perspective," \emph{IEEE Transactions on Wireless Communications}, vol. 13, no. 10, pp. 5340-5351, Oct. 2014.


\bibitem{Paper:Ho} C.K. Ho, D. Yuan, S. Sun, ``Data Offloading in Load Coupled Networks: A Utility Maximization Framework," \emph{IEEE Transactions on Wireless Communications}, vol. 13, no. 4, pp. 1912-1931, April 2014.

\bibitem{Paper:EE_Chen} X. Chen, J. Wu, Y. Cai, H. Zhang, T. Chan, ``Energy-Efficiency Oriented Traffic Offloading in Wireless Networks: A Brief Survey and A Learning Approach for Heterogeneous Cellular Networks," \emph{IEEE Journal on Selected Areas in Communications}, vol. 33, no. 4, April 2015.


\bibitem{Paper:D2D_Lei} L. Lei, Z. Zhong, C. Lin, X. Shen, ``Operator Controlled Device-to-Device Communications in LTE-Advanced networks," \emph{IEEE Wireless Communications}, vol. 19, no. 3,  pp. 96-104, June 2012.

\bibitem{Paper:DoubleAuction_Iosifidis} G. Iosifidis, L. Gao, J. Huang, and L. Tassiulas, ``A Double Auction Mechanism for Mobile Data Offloading Markets," to appear in \emph{IEEE Transactions on Networking}.


\bibitem{Paper:Backhaul_Duan} Y. Yang, T.Q.S. Quek, L. Duan, ``Backhaul-Constrained Small Cell Networks: Refunding and QoS Provisioning," \emph{IEEE Transactions on Wireless Communications}, vol. 13, no. 9, pp. 5148-5161, Sept. 2014.




\bibitem{Paper:Economic_Huang} L. Gao, G. Iosifidis, J. Huang, L. Tassiulas, ``Economics of Mobile Data Offloading," in \emph{Proc of Smart Data Pricing Workshop, IEEE INFOCOM'2013}.


\bibitem{Paper:Delay_Zhuo} X. Zhuo, W. Gao, G. Cao, S. Hua, ``An Incentive Framework for Cellular Traffic Offloading," \emph{IEEE Transactions on Mobile Computing}, vol. 13, no. 3, pp. 541-555, March 2014.

\bibitem{Paper:EconomicGain_Song} J. Lee, Y. Yi, S. Chong, and Y. Jin, ``Economics of WiFi Offloading: Trading Delay for Cellular Capacity," \emph{IEEE Transactions on Wireless Communications}, vol. 13, no. 3, pp. 1540-1544, March 2014.

\bibitem{Paper:AMUSE} Y. Im, C.J. Wong, S. Ha, S. Sen, T. Kwon, M. Chiang, ``AMUSE: Empowering Users for Cost-Aware Offloading with Throughput-Daly Tradeoffs," in \emph{Proc. of IEEE INFOCOM'2013}.

\bibitem{Paper:Delay_Huang2} M.H. Cheung, R. Southwell, J. Huang, ``Congestion-Aware Network Selection and Data Offloading," in \emph{IEEE CISS'2014}.




\bibitem{Paper:Dual} S.C. Jha, K. Sivanesan, R. Vannithamby, A.T. Koc, ``Dual Connectivity in LTE Small Cell Networks," in \emph{Proc. of IEEE GLOBECOM'2014}.


\bibitem{Paper:spectrumsharing1} Nokia Networks, ``Future Work: Optimizing Spectrum Utilisation towards 2020," White Paper, available online at \url{http://networks.nokia.com/file/30301/optimising-spectrum-utilisation-towards-2020}.


\bibitem{Paper:SpectrumSharing_Niyato} S. Guruacharya, D. Niyato, D. I. Kim, E. Hossain, ``Hierarchical Competition for Downlink Power Allocation in OFDMA Femtocell Networks," \emph{IEEE Trans. Wireless Commun.}, vol. 12, no. 4, pp. 1543-1553, April 2013.






\bibitem{Paper:APChannel1} N. Shetty, S. Parekh, J. Walrand, ``Economics of femtocells", \emph{in Proc. of IEEE GLOBECOM'2009}.




\bibitem{Paper:Channel} R. Zhang, ``Optimal Dynamic Resource Allocation for Multi-Antenna Broadcasting with Heterogeneous Delay-Constrained Traffic," \emph{IEEE Journal of Sel. Topics in Signal Processing}, vol.2, no. 2, pp. 243-255, 2008.

\bibitem{Paper:BandwidthAP} O. Bejarano, E.W. Knightly, ``IEEE 802.11ac: From Channelization to Multi-User MIMO," \emph{IEEE Communications Magazine}, pp. 84-90, Oct. 2013.


\bibitem{Paper:PowerNI} National Instruments, ``Introduction to UMTS Device Testing Transmitter and Receiver Measurements for WCDMA Devices," available online at \url{http://download.ni.com/evaluation/rf/Introduction_to_UMTS_Device_Testing.pdf}.

\bibitem{Paper:PowerDBM} S. Trifunovic, A. Picu, T. Hossmann, K.A. Hummel, ``Slicing the Battery Pie: Fair and Efficient Energy Usage in Device-to-Device Communication via Role Switching," in \emph{in Proc. of ACM CHANT'2013}.

\bibitem{Paper:Price} S. Ha, C.J. Wong, S. Sen, M. Chiang, ``Pricing by Timing: Innovating Broadband Data Plans,"
in \emph{Proc. of SPIE OPTO Broadband Access Communication Technologies VI Conference, 2012}.

\bibitem{Paper:TransmitPower_Huang} J. Huang, R. Berry and M. L. Honig, ``Distributed Interference Compensation
for Wireless Networks," \emph{IEEE Journal on Selected Areas in Communications}, vol.24, no.5, pp.1074-1084, May 2006.

\bibitem{Paper:TransmitPower_Mung} M. Chiang, C.W. Tan, D.P. Palomar, D.O. Neill, and D. Julian, ``Power Control by Geometric Programming," \emph{IEEE Transations on Wireless Communications}, vol. 6, no. 7, pp. 2640-2651, July 2007.

\bibitem{Book:Monotonic} Y.J. Zhang, L.P Qian, and J. Huang, ``Monotonic Optimization in Communication and Networking Systems," \emph{Foundation and Trends in Networking}, Now Publisher, October 2013.

\bibitem{Paper:HoangTuy} H. Tuy, ``Monotonic Optimization: Problems and Solution Approaches," \emph{SIAM Journal of Optimization}, vol.11, no.2, 2000.


\bibitem{Paper:MonotonicExample2} L. Qian, Y. J. Zhang, and J. Huang, ``MAPEL: Achieving global optimality for
a non-convex wireless power control problem," \emph{IEEE Transactions on Wireless Communications}, vol. 8, no. 3, pp. 1553-1563, 2009.

\bibitem{Paper:ConcaveMinimization_DSL} Y. Xu, T. Le-Ngoc, S. Panigrahi, ``Global Concave Minimization for Optimal Spectrum Balancing in Multi-User DSL Networks," \emph{IEEE Transactions on Signal Processing}, vol 56, no. 7, pp. 2875-2885, July 2008.


\bibitem{Book:ConcaveMinimization} R. Horst, and H. Tuy, ``Global Optimization: Deterministic Approaches", Springer, 1996.

\bibitem{Paper:ConcaveMinimization} P. M. Pardalos, and J. B. Rosen, ``Methods for Global Concave Minimization: A Bibliographic Survey", \emph{SIAM Review}, vol. 28, no. 3, pp. 367-379, Sept. 1986.



\bibitem{Paper:PowerConvergence} G.J. Foschini, Z. Miljanic, ``A Simple Distributed Autonomous Power Control Algorithm and its Convergence," \emph{IEEE Transactions on Vehicular Technology}, vol. 42, no. 4, Nov. 1993.

\bibitem{Software:LINGO} L. Schrage, ``Optimization Modeling with LINGO," the 5th edition, Lindo System, Jan. 1999.

\end{thebibliography}
\end{document}